\documentclass[a4paper]{amsart}

\usepackage[T1]{fontenc}
\usepackage{amsmath}
\usepackage{subcaption}
\usepackage{amssymb}
\usepackage{shuffle}
\usepackage{tikz}
\usepackage[inline]{enumitem}
\usepackage[disable]{todonotes}
\usetikzlibrary{automata}
\usetikzlibrary{positioning}
\usetikzlibrary{calc}
\usepackage{hyperref}
\usepackage{aliascnt}
\usepackage[capitalize]{cleveref}

\newcommand{\rev}[1]{{#1}^R}

\newcommand{\eword}{\varepsilon}
 
\newcommand{\Powerset}[1]{\mathcal{P}({#1})}
\newcommand{\TrivialMonoid}{\mathbf{1}}
\newcommand{\HG}{\mathsf{G}}
\newcommand{\HF}{\mathsf{F}}

\newcommand{\HomSLI}[1]{\mathsf{SLI}(#1)}
\newcommand{\Alg}[1]{\mathsf{Alg}(#1)}

\newcommand{\Uclosure}[2][YYY]{\ifthenelse{\equal{#1}{YYY}}{#2\mathord{\uparrow}}{#2\mathord{\uparrow}_{#1}}}
\newcommand{\Parikh}[1]{\Psi\left(#1\right)}
\newcommand{\ParikhInv}[1]{\Psi^{-1}(#1)}
\newcommand{\ParikhMap}{\Psi}

\newcommand{\Lang}[2][YYY]{\ifthenelse{\equal{#1}{YYY}}{\mathsf{L}(#2)}{\mathsf{L}_{#1}(#2)}}

\newcommand{\DesImp}[2]{``\ref{#1}~$\Rightarrow$~\ref{#2}''}

\newcommand{\M}[1][]{\mathbb{M}\ifthenelse{\equal{#1}{}}{}{(#1)}}

\newcommand{\congruence}{\equiv}

\newcommand{\cA}{\mathcal{A}}
\newcommand{\cB}{\mathcal{B}}

\newcommand{\C}{\mathcal{C}}
\newcommand{\D}{\mathcal{D}}

\newcommand{\B}{\mathbb{B}}

\newcommand{\Pfour}{\mathsf{P4}}
\newcommand{\Cfour}{\mathsf{C4}}

\newcommand{\RInv}[2][YYY]{\ifthenelse{\equal{#1}{YYY}}{\overrightarrow{\mathsf{I}}(#2)}{\overrightarrow{\mathsf{I}}_{#1}(#2)}}
\newcommand{\LInv}[2][YYY]{\ifthenelse{\equal{#1}{YYY}}{\overleftarrow{\mathsf{I}}(#2)}{\overleftarrow{\mathsf{I}}_{#1}(#2)}}

\newcommand{\Rclass}{R}

\newcommand{\VA}[1]{\mathsf{VA}(#1)}

\newcommand{\VT}[2][YYY]{\ifthenelse{\equal{#1}{YYY}}{\mathsf{VT}(#2)}{\mathsf{VT}(#2, #1)}}
\newcommand{\VTp}[2][YYY]{\ifthenelse{\equal{#1}{YYY}}{\mathsf{VT}^+(#2)}{\mathsf{VT}^+(#2, #1)}}
\newcommand{\Reg}{\mathsf{Reg}}
\newcommand{\CF}[1][YYY]{\ifthenelse{\equal{#1}{YYY}}{\mathsf{CF}}{\mathsf{CF}_{#1}}}
\newcommand{\fiCF}{\mathsf{fiCF}}
\newcommand{\Petri}{\mathsf{P}}

\newcommand{\Trio}[1]{\mathcal{T}(#1)}

\newcommand{\AH}[1][XX]{\mathsf{AH}\ifthenelse{\equal{#1}{XX}}{}{(#1)}}
\newcommand{\RE}[1][XX]{\mathsf{RE}\ifthenelse{\equal{#1}{XX}}{}{(#1)}}
\newcommand{\REC}[1][XX]{\mathsf{REC}\ifthenelse{\equal{#1}{XX}}{}{(#1)}}
\newcommand{\Prio}{\mathsf{Prio}}

\newcommand{\MonDecidable}{\mathsf{DEC}}
\newcommand{\MonRemaining}{\mathsf{REM}}

\newcommand{\MonDecidableSimple}{\mathsf{SC}^{\pm}}
\newcommand{\MonRemainingSimple}{\mathsf{SC^{+}}}

\newcommand{\defeq}{=}

\newcommand{\lautsteps}[2][YYY]{\ifthenelse{\equal{#1}{YYY}}{\xrightarrow{#2}}{\xrightarrow{#2}_{#1}}}
\newcommand{\autstep}[1][YYY]{\ifthenelse{\equal{#1}{YYY}}{\rightarrow}{\rightarrow_{#1}}}
\newcommand{\autsteps}[1][YYY]{\ifthenelse{\equal{#1}{YYY}}{\rightarrow^*}{\rightarrow^*_{#1}}}
\newcommand{\autstepsn}[2][YYY]{\ifthenelse{\equal{#1}{YYY}}{\rightarrow^{#2}}{\rightarrow^{#2}_{#1}}}
\newcommand{\grammarstep}[1][YYY]{\ifthenelse{\equal{#1}{YYY}}{\Rightarrow}{\Rightarrow_{#1}}}
\newcommand{\grammarsteps}[1][YYY]{\ifthenelse{\equal{#1}{YYY}}{\Rightarrow^*}{\Rightarrow^*_{#1}}}
\newcommand{\grammarstepsn}[2][YYY]{\ifthenelse{\equal{#1}{YYY}}{\Rightarrow^{#2}}{\Rightarrow^{#2}_{#1}}}

\newcommand{\N}{\mathbb{N}}

\newcommand{\Z}{\mathbb{Z}}

\newcommand{\zerofill}[1]{0|#1}

\newcommand{\underlyingcfourpfour}[5]{
\begin{tikzpicture}[every circle/.style={}, scale=#1]
\fill (0,0) circle (2pt) node (aa) {}    +(1,0) circle (2pt) node (ab) {}      +(1,1) circle (2pt) node (ac) {}    +(0,1) circle (2pt) node (ad) {};
\draw (aa) node[left=3pt] {#3};
\draw (ab) node[right=3pt] {#5};
\draw (ac) node[right=3pt] {#4};
\draw (ad) node[left=3pt] {#2};
\draw (aa.center) -- (ab.center);
\draw (aa.center) -- (ad.center) -- (ac.center);
\draw (ab.center) -- (ac.center) [densely dotted];
\draw [densely dotted] (aa.center) ++(-135:3pt) circle (3pt);
\draw [densely dotted] (ab.center) ++(-45:3pt) circle (3pt);
\draw [densely dotted] (ac.center) ++(45:3pt) circle (3pt);
\draw [densely dotted] (ad.center) ++(135:3pt) circle (3pt);
\end{tikzpicture}
}

\newcommand{\unloopedcycle}[1]{
\begin{tikzpicture}[every circle/.style={}, scale=#1]
\fill (0,0) circle (2pt) node (a) {}    (1,0) circle (2pt) node (b)  {}   (0,-1) circle (2pt) node (c) {}   (1,-1) circle (2pt) node (d) {};
\draw (a.center) -- (b.center) -- (d.center) -- (c.center) -- (a.center);
\end{tikzpicture}
}

\newcommand{\unloopedpath}[1]{
\begin{tikzpicture}[every circle/.style={}, scale=#1]
\fill (0,0) circle (2pt) node (a) {}    (1,0) circle (2pt) node (b)  {}   (2,0) circle (2pt) node (c) {}   (3,0) circle (2pt) node (d) {};
\draw (a.center) -- (b.center) -- (c.center) -- (d.center);
\end{tikzpicture}
}

\newcommand{\pnpZero}[1]{
\begin{tikzpicture}[every circle/.style={}, scale=#1]
\fill (0,0) circle (2pt) node (a) {}    (1,0) circle (2pt) node (b)  {}   (2,0) circle (2pt) node (c) {};
\draw (a.center) -- (b.center) -- (c.center);
\end{tikzpicture}
}

\newcommand{\pnpOne}[1]{
\begin{tikzpicture}[every circle/.style={}, scale=#1]
\fill (0,0) circle (2pt) node (a) {}    (1,0) circle (2pt) node (b)  {}   (2,0) circle (2pt) node (c) {};
\draw (a.center) -- (b.center) -- (c.center);
\draw (a.center) ++(90:3pt) circle (3pt);
\end{tikzpicture}
}

\newcommand{\pnpTwo}[1]{
\begin{tikzpicture}[every circle/.style={}, scale=#1]
\fill (0,0) circle (2pt) node (a) {}    (1,0) circle (2pt) node (b)  {}   (2,0) circle (2pt) node (c) {};
\draw (a.center) -- (b.center) -- (c.center);
\draw (a.center) ++(90:3pt) circle (3pt);
\draw (c.center) ++(90:3pt) circle (3pt);
\end{tikzpicture}
}

\newlist{conditions}{enumerate}{2}
\setlist[conditions,1]{label=\arabic*.,ref=\arabic*}
\setlist[conditions,2]{label=\emph{\alph*}),ref=\theconditionsi\alph*}
\crefname{conditionsi}{condition}{conditions}
\crefname{conditionsii}{condition}{conditions}

\newlist{intextconditions}{enumerate*}{1}
\setlist[intextconditions,1]{label=(\roman*)}

\newlist{intextenum}{enumerate*}{1}
\setlist[intextenum,1]{label=\roman*)}

\newlist{openproblems}{enumerate}{1}
\setlist[openproblems,1]{label=\arabic*.}

\newlist{deflist}{itemize}{1}
\newlist{algorithms}{itemize}{1}
\newlist{generalizedresults}{itemize}{1}
\newlist{tasks}{itemize}{1}
\newlist{approaches}{itemize}{1}
\newlist{proofcases}{itemize}{1}
\newlist{researchdirections}{itemize}{1}
\newlist{results}{itemize}{1}
\newlist{instructions}{itemize}{1}
\newlist{contributions}{description}{1} 
\setlist[deflist,algorithms,generalizedresults,tasks,approaches,proofcases,researchdirections,results,instructions,contributions]{label=\textbullet}

\newlist{operations}{enumerate}{1}
\setlist[operations,1]{label=(\roman*)}

\newlist{rules}{enumerate}{1}
\setlist[rules,1]{label=(\Roman*)}

\newlist{phases}{enumerate}{1}
\setlist[phases,1]{label=(\alph*)}
\crefname{phasesi}{phase}{phases}
\Crefname{phasesi}{Phase}{Phases}

\newcommand{\arev}[1]{\todo[author=R1]{#1}}

\newtheorem{thm}{Theorem}[section]

\newaliascnt{lem}{thm}
\newtheorem{lem}[lem]{Lemma}
\aliascntresetthe{lem}

\newaliascnt{prop}{thm}
\newtheorem{prop}[prop]{Proposition}
\aliascntresetthe{prop}

\theoremstyle{remark}
\newaliascnt{rmk}{thm}
\newtheorem{rmk}[rmk]{Remark}
\aliascntresetthe{rmk}
\crefname{rmk}{Remark}{Remarks}

\newcommand{\myqed}{}

\crefname{thm}{Theorem}{Theorems}
\crefname{lem}{Lemma}{Lemmas}
\crefname{prop}{Proposition}{Propositions}
\Crefname{thm}{Theorem}{Theorems}
\Crefname{lem}{Lemma}{Lemmas}
\Crefname{prop}{Proposition}{Propositions}

\begin{document}

\title{The Emptiness Problem for Valence Automata over Graph Monoids}
\author{Georg Zetzsche}
\email{zetzsche@lsv.fr}
\address{LSV, CNRS \& ENS Paris-Saclay, France}
\thanks{The author is supported by a fellowship within the Postdoc-Program of the German Academic Exchange Service (DAAD) and by
Labex DigiCosme, Univ.\ Paris-Saclay, project VERICONISS.}

\begin{abstract}
  This work studies which storage mechanisms in automata permit
  decidability of the emptiness problem. The question is formalized
  using valence automata, an abstract model of automata in which the
  storage mechanism is given by a monoid. For each of a variety of
  storage mechanisms, one can choose a (typically infinite) monoid $M$
  such that valence automata over $M$ are equivalent to (one-way)
  automata with this type of storage.  In fact, many important
  storage mechanisms can be realized by monoids defined by finite
  graphs, called graph monoids.  Examples include pushdown stacks,
  partially blind counters (which behave like Petri net places), blind
  counters (which may attain negative values), and combinations
  thereof.

  Hence, we study for which graph monoids the emptiness problem for
  valence automata is decidable. A particular model realized by graph
  monoids is that of Petri nets with a pushdown stack.  For these,
  decidability is a long-standing open question and we do not answer
  it here.

  However, if one excludes subgraphs corresponding to this model, a
  characterization can be achieved.
  Moreover, we provide a description of those storage mechanisms for
  which decidability remains open.  This leads to a model that naturally
  generalizes both pushdown Petri nets and the priority multicounter
  machines introduced by Reinhardt.

  The cases that are proven decidable constitute a natural and
  apparently new extension of Petri nets with decidable
  reachability. It is finally shown that this model can be combined
  with another such extension by Atig and Ganty: We present a further
  decidability result that subsumes both of these Petri net
  extensions.
\end{abstract}

\maketitle

\section{Introduction}

For each storage mechanism in one-way automata, it is an important
question whether the emptiness problem is decidable. It therefore
seems prudent to aim for general insights into which properties of 
storage mechanisms are responsible for decidability or undecidability.

Our approach to obtain such insights is the model of valence
automata. These feature a finite-state control and a (typically
infinite) monoid that represents a storage mechanism. The edge
inscriptions consist of an input word and an element of the
monoid. Then, a computation is accepting if it arrives in a final
state and composing the encountered monoid elements yields the neutral
element. This way, by choosing a suitable monoid, one can realize a
variety of storage mechanisms. Hence, our question becomes: \emph{For
  which monoids $M$ is the emptiness problem for valence automata over
  $M$ decidable?}

We address this question for a class of monoids that was introduced
in~\cite{Zetzsche2013a} and accommodates a number of storage
mechanisms that have been studied in automata theory. Examples include
\emph{pushdown stacks}, \emph{partially blind counters} (which behave
like Petri net places), and \emph{blind counters} (which may attain
negative values; these are in most situations interchangeable with
reversal-bounded counters), and combinations thereof. See
\cite{Zetzsche2016d,Zetzsche2016c} for an overview. These monoids are
defined by graphs and thus called \emph{graph monoids}\footnote{They
  are not to be confused with the closely related, but different
  concept of \emph{trace monoids}~\cite{DiekertRozenberg1995},
  i.e. monoids of Mazurkiewicz traces, which some authors also call
  graph monoids.}.

A particular type of storage mechanism that can be realized by graph monoids
are partially blind counters that can be used simultaneously with a pushdown
stack. Automata with such a storage are equivalent to \emph{pushdown Petri nets
(PPN)}, i.e. Petri nets where the transitions can also operate on a pushdown
stack. This means, a complete characterization of graph monoids with a
decidable emptiness problem would entail an answer to the long-standing open
question of whether reachability is decidable for this Petri net
extension~\cite{Reinhardt2008}. Partial solutions have recently been obtained
by Atig and Ganty~\cite{AtigGanty2011} and by Leroux, Sutre, and
Totzke~\cite{LerouxSutreTotzke2015}.

\paragraph{Contribution} While this work does not answer this open
question concerning PPN, it does provide a characterization among all
graph monoids that avoid this elusive storage type. More precisely, we
identify a set of graphs, `PPN-graphs', each of which corresponds
precisely to PPN with one Petri net place. Then, among all graphs
$\Gamma$ avoiding PPN-graphs as induced subgraphs, we characterize
those for which the graph monoid $\M\Gamma$ results in a decidable
emptiness problem.  Furthermore, we provide a simple, more mechanical
(as opposed to algebraic) description of
\begin{enumerate}[label=(\roman*)]
\item\label{decidable:mechanism} the storage mechanism emerging as the most
general decidable case and 
\item\label{open:mechanism} a type of mechanism equivalent to the cases we leave open.
\end{enumerate}
The model \labelcref{decidable:mechanism} is a new extension of partially blind
counter automata (i.e. Petri nets). While the decidability proof employs a
reduction to Reinhardt's priority multicounter machines~\cite{Reinhardt2008},
the model \labelcref{decidable:mechanism} seems to be expressively incomparable
to Reinhardt's model. The model \labelcref{open:mechanism} is a class of
mechanisms whose simplest instance are the pushdown Petri nets and which also
naturally subsumes priority multicounter machines (see also \cref{ppn-prio}).

Another recent extension of the decidability of reachability of Petri nets has
been obtained by Atig and Ganty~\cite{AtigGanty2011}. In fact, it is a partial
solution to the reachability problem for PPN. Their proof also relies on
priority multicounter machines. They show that given a \emph{finite-index}
context-free language $K$ and a language $L$ generated by a Petri net, it is
decidable whether the intersection $K\cap L$ is empty. Note that without the
finite-index requirement, this would be equivalent to the reachability problem
for PPN. Our final contribution is a decidability result that subsumes both the
decidability of model \labelcref{decidable:mechanism} and the result of Atig
and Ganty. We present a natural language class that contains both the
intersections considered by Atig and Ganty and the languages of model
\labelcref{decidable:mechanism} and still has a decidable emptiness problem. To
this end, we employ a slightly stronger (and perhaps simpler) version of Atig
and Ganty's reduction.

Hence, the perspective of valence automata allows us to identify
natural storage mechanisms
that \begin{enumerate*}[label=(\roman*)] \item push the frontier of
decidable emptiness (and hence reachability) and \item let us
naturally interpret PPN and priority multicounter machines as
special cases of a more powerful model that might enjoy
decidability\end{enumerate*}, respectively.

The paper is structured as follows. We present the main results in
\cref{results} and prove them in
\cref{proofundecidability,proofmain,proofexpressive}.  \Cref{proofundecidability}
presents the undecidability part, \cref{proofmain} treats the decidable cases,
and \cref{proofexpressive} shows the expressive equivalence with the more
mechanical descriptions.  In \cref{synthesis}, we present the enhanced
decidability result that also subsumes the one by Atig and Ganty.

This work is an extended version of the
paper~\cite{Zetzsche2015c}. This version provides proofs of the
results of \cite{Zetzsche2015c} and the enhanced decidability result.
Moreover, it contains proofs of some results that first appeared
in~\cite{Zetzsche2013a,Zetzsche2015a}, but have not yet undergone
journal peer review.

\section{Preliminaries}\label{preliminaries}
A \emph{monoid} is a set $M$ together with a binary associative
operation such that $M$ contains a neutral element. Unless the monoid
at hand warrants a different notation, we will denote the neutral
element by $1$ and the product of $x,y\in M$ by $xy$. If $X$ is a set
of symbols, $X^*$ denoted the set of words over $X$. The length of the
word $w\in X^*$ is denoted $|w|$. An \emph{alphabet} is a finite set
of symbols. The empty word is denoted by $\eword\in X^*$. Let
$P\subseteq X\times X$ is a set of pairs of symbols, then the
\emph{semi-Dyck language} over $P$, denoted $D_P^*$ is the smallest
subset of $X^*$ such that $\eword\in D_P^*$ and whenever
$uv\in\D_P^*$, then also $ua\bar{a}v\in D_P^*$ for every
$(a,\bar{a})\in P$. If $P=\{(a_i,\bar{a}_i)\mid i\in\{1,\ldots,n\}\}$,
then we also write $D_n^*$ instead of $D_P^*$. Moreover, if
$P=\{(a,b)\}$, then the words in $D_P^*$ are called \emph{semi-Dyck
  words over $a,b$}. If $w\in X^*$ is a word with $w=x_1\cdots x_n$
for $x_1,\ldots,x_n\in X$, then $w^R$ denotes $w$ in reverse,
i.e. $w^R=x_n\cdots x_1$.

For an alphabet $X$ and languages $L,K\subseteq X^*$, the \emph{shuffle
product} $L\shuffle K$ is the set of all words $u_0 v_1 u_1 \cdots v_n
u_n$ where $u_0,\ldots,u_n,v_1,\ldots,v_n\in X^*$, $u_0\cdots u_n\in L$,
and $v_1\cdots v_n\in K$. For a subset $Y\subseteq X$, we define the
\emph{projection morphism} $\pi_Y\colon X^*\to Y^*$ by $\pi_Y(y)=y$ for
$y\in Y$ and $\pi_Y(x)=\eword$ for $x\in X\setminus Y$. Moreover, we define
$|w|_Y=|\pi_Y(w)|$ and for $x\in X$, we set $|w|_x=|w|_{\{x\}}$.

\paragraph{Valence automata} As a framework for studying which storage
mechanisms permit decidability of the emptiness problem, we employ valence
automata.  They feature a monoid that dictates which computations are valid.
Hence, by an appropriate choice of the monoid, valence automata can be
instantiated to be equivalent to a concrete automata model with storage.  For
the purposes of this work, \emph{equivalent} is meant with respect to accepted
languages. Therefore, we regard valence automata as language accepting devices.

Let $M$ be a monoid and $X$ an alphabet.
A \emph{valence automaton over $M$} is a tuple $\cA=(Q,X,M,E,q_0,F)$, in
which
\begin{intextconditions}
\item $Q$ is a finite set of \emph{states},
\item $E$ is a finite subset of $Q\times X^*\times M\times Q$, called the set of \emph{edges},
\item $q_0\in Q$ is the \emph{initial state}, and
\item $F\subseteq Q$ is the set of \emph{final states}.
\end{intextconditions}
For $q,q'\in Q$, $w,w'\in X^*$, and $m,m'\in M$, we write $(q,w,m)
\autstep[\cA] (q',w',m')$ if there is an edge $(q,v,n,q')\in E$ such
that $w'=wv$ and $m'=mn$. The language \emph{accepted} by $\cA$ is
then 
\[ \Lang{\cA}\defeq\{w\in X^* \mid (q_0,\eword,1)\autsteps[\cA]
(f,w,1)~\text{for some $f\in F$} \}. \]
The class of languages accepted by valence automata over $M$ is
denoted by $\VA{M}$. If $\mathcal{M}$ is a class of monoids, we write
$\VA{\mathcal{M}}$ for $\bigcup_{M\in\mathcal{M}} \VA{M}$.

\paragraph{Graphs} A \emph{graph} is a pair $\Gamma=(V,E)$ where $V$ is a
finite set and $E$ is a subset of $\{S\subseteq V \mid 1\le|S|\le 2\}$. The
elements of $V$ are called \emph{vertices} and those of $E$ are called
\emph{edges}. Vertices $v,w\in V$ are \emph{adjacent} if $\{v,w\}\in E$.  If
$\{v\}\in E$ for some $v\in V$, then $v$ is called a \emph{looped} vertex,
otherwise it is \emph{unlooped}. A \emph{subgraph} of $\Gamma$ is a graph
$(V',E')$ with $V'\subseteq V$ and $E'\subseteq E$. Such a subgraph is called
\emph{induced (by $V'$)} if $E'=\{ S\in E \mid S\subseteq V'\}$, i.e. $E'$
contains all edges from $E$ incident to vertices in $V'$. By
$\Gamma\setminus\{v\}$, for $v\in V$, we denote the subgraph of $\Gamma$
induced by $V\setminus \{v\}$. By $\Cfour{}$ ($\Pfour{}$), we denote a graph
that is a cycle (path) on four vertices;
see~\cref{basics:graphmonoids:fig}.
Moreover, $\Gamma^-$ denotes the graph obtained from
$\Gamma$ by deleting all loops: We have $\Gamma^-=(V,E^-)$, where $E^-=\{S\in E
\mid |S|=2\}$. The graph $\Gamma$ is \emph{loop-free} if $\Gamma^-=\Gamma$.
Finally, a \emph{clique} is a loop-free graph in which any two distinct
vertices are adjacent. 

\begin{figure}[t]
\begin{subfigure}[b]{0.3\textwidth}\centering
\unloopedpath{1}
\caption{$\Pfour$\label{basics:graphmonoids:fig:pfour}}
\end{subfigure}
\begin{subfigure}[b]{0.3\textwidth}\centering
\unloopedcycle{1}
\caption{$\Cfour{}$\label{basics:graphmonoids:fig:cfour}}
\end{subfigure}
\caption{Graphs $\Cfour{}$ and $\Pfour{}$.\label{basics:graphmonoids:fig}}
\end{figure}

\paragraph{Products and presentations}
If $M$, $N$ are monoids, then $M\times N$ denotes their \emph{direct
product}, whose set of elements is the cartesian product of $M$ and $N$
and composition is defined component-wise. By $M^n$, we denote the
$n$-fold direct product, i.e.  $M\times\cdots\times M$ with $n$
factors.  

Let $A$ be a (not necessarily finite) set of symbols and $R$ be a
subset of $A^*\times A^*$. The pair $(A,R)$ is called a \emph{(monoid)
  presentation}. The smallest congruence of the free monoid $A^*$
containing $R$ is denoted by $\congruence_R$ and we will write $[w]_R$
for the congruence class of $w\in A^*$. The \emph{monoid presented by
  $(A,R)$} is defined as $A^*/\mathord{\congruence_R}$. Note that
since we did not impose a finiteness restriction on $A$, up to
isomorphism, every monoid has a presentation.  If
$A=\{a_1,\ldots,a_n\}$ and $R=\{(r_i,\bar{r}_i)\mid
i\in\{1,\ldots,k\}\}$, we also use the shorthand $\langle
a_1,\ldots,a_n \mid r_1=\bar{r}_1,\ldots, r_k=\bar{r}_k\rangle$ to
denote the monoid presented by $(A,R)$.

Furthermore, for
monoids $M_1$, $M_2$ we can find presentations $(A_1,R_1)$ and
$(A_2,R_2)$ such that $A_1\cap A_2=\emptyset$. We define the
\emph{free product} $M_1*M_2$ to be presented by $(A_1\cup A_2,
R_1\cup R_2)$.  Note that $M_1*M_2$ is well-defined up to isomorphism.
In analogy to the $n$-fold direct product, we write $M^{(n)}$ for the
$n$-fold free product of $M$.

\paragraph{Graph monoids}
A presentation $(A,R)$ in which $A$ is a finite alphabet is a
\emph{Thue system}.
 To each graph $\Gamma=(V,E)$, we associate the Thue system
$T_\Gamma=(X_\Gamma, R_\Gamma)$ over the alphabet $X_\Gamma=\{a_v,
\bar{a}_v \mid v\in V\}$. $R_\Gamma$ is defined as
\[ R_\Gamma = \{(a_v\bar{a}_v, \eword) \mid v\in V \} \cup \{ (xy,yx) \mid x\in \{a_v, \bar{a}_v\}, ~y\in \{a_w,\bar{a}_w\}, ~ \{v,w\}\in E\}. \]
In particular, we have $(a_v\bar{a}_v, \bar{a}_va_v)\in R_\Gamma$
whenever $\{v\}\in E$. To simplify notation, the congruence
$\congruence_{R_\Gamma}$ is then also denoted by $\congruence_\Gamma$. We are now ready to define
graph monoids. To each graph $\Gamma$, we associate the monoid
\[ \M\Gamma ~~=~~ X^*_\Gamma/\mathord{\congruence_\Gamma}. \]
The monoids of the form $\M\Gamma$ are called \emph{graph monoids}.

\paragraph{Storage mechanisms as graph monoids}
Let us briefly discuss how to realize storage mechanisms by graph
monoids. First, suppose $\Gamma_0$ and $\Gamma_1$ are disjoint
graphs. If $\Gamma$ is the union of $\Gamma_0$ and $\Gamma_1$, then
$\M\Gamma\cong\M\Gamma_0*\M\Gamma_1$ by definition. Moreover, if
$\Gamma$ is obtained from $\Gamma_0$ and $\Gamma_1$ by drawing an edge
between each vertex of $\Gamma_0$ and each vertex of $\Gamma_1$, then
$\M\Gamma\cong\M\Gamma_0\times\M\Gamma_1$.

If $\Gamma$ consists of one vertex $v$ and has no edges,
the only rule in the Thue system is $(a_v\bar{a}_v,\varepsilon)$.  In this
case, $\M\Gamma$ is also denoted as $\B$ and we will refer to it as the
\emph{bicyclic monoid}.  The generators $a_v$ and $\bar{a}_v$ are then also
written $a$ and $\bar{a}$, respectively. It is not hard to see that $\B$
corresponds to a \emph{partially blind counter}, i.e. one that attains only
non-negative values and has to be zero at the end of the computation.
Moreover, if $\Gamma$ consists of one looped vertex, then $\M\Gamma$ is
isomorphic to $\Z$ and thus realizes a \emph{blind counter}, which can go below
zero and is zero-tested in the end.

If one storage mechanism is realized by a monoid $M$, then the monoid
$\B*M$ corresponds to the mechanism that \emph{builds stacks}: A
configuration of this new mechanism consists of a sequence
$c_0ac_1\cdots ac_n$, where $c_0,\ldots,c_n$ are configurations of the
mechanism realized by $M$. We interpret this as a stack with the
entries $c_0,\ldots,c_n$. One can open a new stack entry on top (by
multiplying $a\in\B$), remove the topmost entry if empty (by
multiplying $\bar{a}\in\B$) and operate on the topmost entry using the
old mechanism (by multiplying elements from $M$). In particular,
$\B*\B$ describes a pushdown stack with two stack
symbols. See~\cite{Zetzsche2016d} for more examples
and~\cite{Zetzsche2016c} for more details.

As a final example, suppose $\Gamma$ is one edge short of being a clique,
then $\M\Gamma\cong \B^{(2)}\times \B^{n-2}$, where $n$ is the number of
vertices in $\Gamma$.  Then, by the observations above, valence automata over
$\M\Gamma$ are equivalent to Petri nets with $n-2$ unbounded places and access to
a pushdown stack. Hence, for our purposes, a \emph{pushdown Petri
net} is a valence automaton over $\B^{(2)}\times\B^n$ for some $n\in\N$.

\section{Results}\label{results}
As a first step, we exhibit graphs $\Gamma$ for which $\VA{\M\Gamma}$ includes
the recursively enumerable languages. 
\begin{thm}\label{re:pfourcfour}
Let $\Gamma$ be a graph such that $\Gamma^-$ contains $\Cfour$ or $\Pfour$ as
an induced subgraph.  Then $\VA{\M\Gamma}$ is the class of recursively
enumerable languages. In particular, the emptiness problem is undecidable for
valence automata over $\M\Gamma$.
\end{thm}
This unifies and slightly strengthens a few undecidability results
concerning valence automata over graph monoids. The case that all
vertices are looped was shown by Lohrey and
Steinberg~\cite{LohreySteinberg2008} (see also the discussion of
\cref{ratmp:graphgroups}). Another case appeared in
\cite{Zetzsche2013a}. We prove \cref{re:pfourcfour} in
\cref{proofundecidability}.

It is not clear whether \cref{re:pfourcfour} describes all $\Gamma$ for
which $\VA{\M\Gamma}$ exhausts the recursively enumerable languages. 
For
example, as mentioned above, if $\Gamma$ is one edge short of being a clique,
then valence automata over $\M\Gamma$ are pushdown Petri nets.
In particular, the emptiness problem for valence automata is equivalent to the
reachability problem of this model, for which decidability is a long-standing
open question~\cite{Reinhardt2008}. In fact, it is already open whether
reachability is decidable in the case of $\B^{(2)}\times\B$, although
Leroux, Sutre, and Totzke have recently made progress on this
case~\cite{LerouxSutreTotzke2015}. Therefore, characterizing those $\Gamma$
with a decidable emptiness problem for valence automata over $\M\Gamma$ would
very likely settle these open questions\footnote{Strictly speaking, it is conceivable
that there is a decision procedure for each $\B^{(2)}\times\B^n$, but no
uniform one that works for all $n$. However, this seems unlikely.}.

However, we will show that if we steer clear of pushdown Petri nets,
we can achieve a characterization. More precisely, we will present
a set of graphs that entail the behavior of pushdown Petri nets.
Then, we show that among those graphs that do not contain these as
induced subgraphs, the absence of $\Pfour{}$ and $\Cfour{}$ already
characterizes decidability.

\newcommand{\PNP}{PPN}
\paragraph{\PNP-graphs}
A graph $\Gamma$ is said to be a \emph{\PNP-graph} if it is isomorphic to one
of the following three graphs:
\begin{align*}
\pnpZero{1} && \pnpOne{1} && \pnpTwo{1}
\end{align*}
We say that the graph $\Gamma$ is \emph{\PNP-free} if it has no
\PNP-graph as an induced subgraph. Observe that a graph $\Gamma$ is
\PNP-free if and only if in the neighborhood of each unlooped vertex,
any two vertices are adjacent.

Of course, the abbreviation `\PNP' refers to `pushdown Petri
nets'. This is justified by the following fact. It is proven in
\cref{proofmain} (page~\pageref{proof-pnpgraphs}).
\begin{prop}\label{pnpgraphs}
If $\Gamma$ is a \PNP-graph, then $\VA{\M\Gamma}=\VA{\B^{(2)}\times \B}$.
\end{prop}

\paragraph{Transitive forests}\label{def:transitiveforest}
In order to exploit the absence of $\Pfour{}$ and $\Cfour{}$ as induced
subgraphs, we will employ a characterization of such graphs as transitive
forests. The \emph{comparability graph} of a tree $t$ is a simple graph with
the same vertices as $t$, but has an edge between two vertices whenever one is
a descendant of the other in $t$. A graph $\Gamma$ is a \emph{transitive
forest} if the simple graph $\Gamma^-$ is a disjoint union of comparability
graphs of trees.  For an example of a transitive forest, see
\cref{example:transitiveforest}.

\begin{figure}[t]
\begin{center}
\begin{tikzpicture}[level distance=1cm]
\tikzstyle{level 1}=[sibling distance=2cm]
\tikzstyle{level 2}=[sibling distance=1cm]
\coordinate (a)
        child  {coordinate (b) [fill] circle (2pt)
                child {coordinate (c) [fill] circle (2pt)}
                child {coordinate (d) [fill] circle (2pt)}
        }
        child {coordinate (e) [fill] circle (2pt)
                child {coordinate (f) [fill] circle (2pt)
                        child {coordinate (g) [fill] circle (2pt)}
                        child {coordinate (h) [fill] circle (2pt)}
                }
        };
\draw (a) [fill] circle (2pt);
\draw [densely dotted] (a) to [bend right] (c);
\draw [densely dotted] (a) to (d);
\draw [densely dotted] (a) to  (f);
\draw [densely dotted] (a) to [out=0, in=45] (h);
\draw [densely dotted] (a) to (g);
\draw [densely dotted] (e) to [out=225, in=90] (g); 
\draw [densely dotted] (e) to [out=315, in=90] (h); 

\coordinate (x) at (3,0)
        child {coordinate (y) [fill] circle (2pt)}
        child {coordinate (z) [fill] circle (2pt)};
\draw (x) [fill] circle (2pt);

\end{tikzpicture}
\end{center}
\caption{Example of a transitive forest. The solid edges are part of
  the trees whose comparability graphs make up the graph. The
  transitive forest consists of both the solid and the dashed edges.}
\label{example:transitiveforest}
\end{figure}
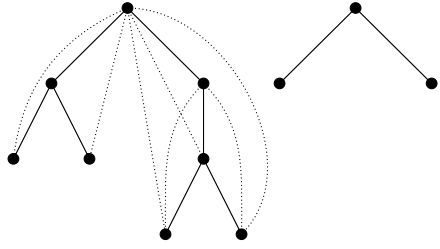

Let $\MonDecidable$ denote the smallest isomorphism-closed class of
monoids such that
\begin{conditions}
\item for each $n\ge 0$, we have $\B^n\in\MonDecidable$ and
\item for $M,N\in\MonDecidable$, we also have $M*N\in\MonDecidable$ and
$M\times\Z\in\MonDecidable$.
\end{conditions}

Our main result characterizes those \PNP-free $\Gamma$ for which
valence automata over $\M\Gamma$ have a decidable emptiness problem.
\begin{thm}\label{emptiness:graphmonoids:equivalence}
Let $\Gamma$ be \PNP-free. Then the following conditions are equivalent:
\begin{conditions}
\item\label{dec:cond:decidable} Emptiness is decidable for valence automata over $\M\Gamma$.
\item\label{dec:cond:cfourpfour} $\Gamma^-$ contains neither $\Cfour{}$ nor $\Pfour{}$ as an induced subgraph.
\item\label{dec:cond:transitive} $\Gamma$ is a transitive forest.
\item\label{dec:cond:monclass} $\M\Gamma\in\MonDecidable$.
\end{conditions}
\end{thm}
We present the proof in \cref{proofmain}.
Note that this generalizes the fact that emptiness is decidable for pushdown
automata (i.e. graphs with no edges) and partially blind multicounter automata
(i.e. cliques), or equivalently, reachability in Petri nets.

Note that if $\Gamma$ has a loop on every vertex, then $\M\Gamma$
is a group. Groups that arise in this way are called \emph{graph
groups}. In general, if a monoid $M$ is a group, then emptiness
for valence automata over $M$ is decidable if and only if the
\emph{rational subset membership problem} is decidable for
$M$~\cite{KambitesSilvaSteinberg2007}. The latter problem asks,
given a rational set $R$ over $M$ and an element $m\in M$, whether
$m\in R$; see~\cite{Lohrey2015a} for more information. Therefore,
\cref{emptiness:graphmonoids:equivalence} extends the following result
of Lohrey and Steinberg~\cite{LohreySteinberg2008},
which characterizes those graph groups for which the rational subset
membership problem is decidable.
\begin{thm}[Lohrey and Steinberg~\cite{LohreySteinberg2008}]\label{ratmp:graphgroups}
Let $\Gamma$ be a graph in which every vertex is looped. Then the rational
subset membership problem for the group $\M\Gamma$ is decidable if and only if
$\Gamma$ is a transitive forest.
\end{thm}
Lohrey and Steinberg show decidability by essentially proving that in
their case, the languages in $\VA{\M\Gamma}$ have semilinear Parikh
images (although they use different terminology). Here, we extend this
argument by showing that in the equivalent cases of
\cref{emptiness:graphmonoids:equivalence}, the Parikh images of
$\VA{\M\Gamma}$ are those of languages accepted by priority
multicounter machines. The latter were introduced and shown to have a
decidable reachability problem by Reinhardt~\cite{Reinhardt2008}.

\paragraph{Intuition for decidable cases} In order to provide an intuition for
those storage mechanisms (not containing a pushdown Petri net) with a
decidable emptiness problem, we present an equally expressive class of monoids
for which the corresponding storage mechanisms are easier to grasp. Let
$\MonDecidableSimple$ be the smallest
isomorphism-closed class of monoids with 
\begin{conditions}
\item for each $n\in\N$, we have $\B^n\in\MonDecidableSimple$,
\item for each $M\in\MonDecidableSimple$, we also have $\B*M\in\MonDecidableSimple$ and $M\times\Z\in\MonDecidableSimple$.
\end{conditions}
Thus, $\MonDecidableSimple$ realizes those storage mechanisms that can be
constructed from a finite set of \emph{partially blind counters} ($\B^n$) by
\emph{building stacks} ($M\mapsto \B*M$) and \emph{adding blind counters}
($M\mapsto M\times\Z$). Then, in fact, the monoids in $\MonDecidableSimple$ produce the same
languages as those in $\MonDecidable$.
\begin{prop}\label{emptiness:decidable:intuition}
$\VA{\MonDecidable}=\VA{\MonDecidableSimple}$.
\end{prop}
\Cref{emptiness:decidable:intuition} is proven in
\cref{proofexpressive}.  While our decidability proof for
$\MonDecidableSimple$ will be a reduction to priority multicounter
machines (see \cref{def:prio} for a definition), it seems likely that
these two models are incomparable in terms of expressiveness (see the
remarks after \cref{prio:decidable}).

\paragraph{Intersections with finite-index languages} This work
exhibits valence automata over $\MonDecidableSimple$ as an extension
of Petri nets that features a type of stack but retains decidability
of the emptiness problem.  Another recent result of this kind has been
obtained by Atig and Ganty~\cite{AtigGanty2011}.  They showed that
given a finite-index context-free language $K$ and a Petri net
language $L$, it is decidable whether $K\cap L$ is empty. Moreover,
they also employ a reduction to priority multicounter machines.  This
raises the question of how the two results relate to each other.  In
\cref{synthesis}, we present a natural language class that subsumes
both the languages of Atig and Ganty and those of
$\VA{\MonDecidableSimple}$ and prove that emptiness is still
decidable. Intuitively, this class is obtained by taking languages of
Atig and Ganty and then applying operators corresponding to
\emph{building stacks} and \emph{adding blind counters}.  The precise
definition and the result can be found in \cref{synthesis}.

\paragraph{Intuition for open cases} We also want to provide an intuition for
the remaining storage mechanisms, i.e. those defined by monoids $\M\Gamma$
about which \cref{re:pfourcfour,emptiness:graphmonoids:equivalence} make no
statement. To this end, we describe a class of monoids that are expressively
equivalent to these remaining cases. The remaining cases are given by those
graphs $\Gamma$ where $\Gamma^-$ does not contain $\Cfour{}$ or $\Pfour{}$, but
$\Gamma$ contains a \PNP{}-graph. Let $\MonRemaining$ denote the class of
monoids $\M\Gamma$, where $\Gamma$ is such a graph. Let $\MonRemainingSimple$
be the smallest isomorphism-closed class of monoids with
\begin{conditions}
\item $\B^{(2)}\times\B\in\MonRemainingSimple$ and 
\item for each $M\in\MonRemainingSimple$, we also have
  $\B*M\in\MonRemainingSimple$ and $M\times\B\in\MonRemainingSimple$.
\end{conditions}
This means, $\MonRemainingSimple$ realizes those storage mechanisms that
are obtained from \emph{a pushdown stack, together with one partially
blind counter} ($\B^{(2)}\times \B$) by the transformations of \emph{building stacks}
($M\mapsto \B*M$) and \emph{adding partially blind counters} ($M\mapsto
M\times\B$). 
\begin{prop}\label{emptiness:remaining:intuition}
$\VA{\MonRemaining} = \VA{\MonRemainingSimple}$.
\end{prop}
We prove \cref{emptiness:remaining:intuition} in
\cref{proofexpressive}.  Of course, $\MonRemainingSimple$ generalizes
pushdown Petri nets, which correspond to monoids $\B^{(2)}\times\B^n$
for $n\in\N$.  Moreover, $\MonRemainingSimple$ also subsumes priority
multicounter machines (see p.~\pageref{def:prio} for a definition) in
a straightforward way: Every time we build stacks, we can use the new
pop operation to realize a zero test on all the counters we have added
so far.  Let $M_0=\TrivialMonoid$ and
$M_{k+1}=\B*(M_k\times\B)$. Then, priority $k$-counter machines
correspond to valence automata over $M_k$ where the stack heights
never exceed $1$.

\begin{rmk}\label{ppn-prio}
  Priority multicounter machines are already subsumed by pushdown
  Petri nets alone: Atig and Ganty~\cite[Lemma 7]{AtigGanty2011} show
  implicitly that for each priority multicounter machine, one can
  construct a pushdown Petri net that accepts the same language.
  Hence, valence automata over $\MonRemainingSimple$ are not the first
  perhaps-decidable generalization of both pushdown Petri nets and
  priority multicounter machines, but they generalize both in a
  natural way.
\end{rmk}

\section{Undecidability}\label{proofundecidability}
In this section, we prove \cref{re:pfourcfour}.  It should be mentioned that a
result similar to \cref{re:pfourcfour} was shown by Lohrey and
Steinberg~\cite{LohreySteinberg2008}: They proved that if every vertex in
$\Gamma$ is looped and $\Gamma^-$ contains $\Cfour$ or $\Pfour$ as an induced
subgraph, then the rational subset membership problem is undecidable for
$\M\Gamma$. Their proof adapts a construction of Aalbersberg and
Hoogeboom~\cite{AalbersbergHoogeboom1989}, which shows that the disjointness
problem for rational sets of traces is undecidable when the independence
relation has $\Pfour$ or $\Cfour$ as an induced subgraph.  An inspection of the
proof presented here, together with its prerequisites
(\cref{re:anbn,afl:trio}), reveals that the employed ideas are very similar to
the combination of Lohrey and Steinberg's and Aalbersberg and Hoogeboom's
proof.

A \emph{language class} is a collection of languages that contains at
least one non-empty language. In this work, for each language class,
there is a way to finitely represent each member of the class.
Moreover, an inclusion $\C\subseteq \D$ between language classes $\C$
and $\D$ is always meant to be \emph{effective}, in other words: Given
a representation of a language in $\C$, we can compute a
representation of that language in $\D$. The same holds for equalities
between language classes.

Let $X$ and $Y$ be alphabets. A relation $T\subseteq X^*\times Y^*$ is
called a \emph{rational transduction} if there is an alphabet $W$, a
regular language $R\subseteq W^*$, and morphisms $g\colon W^*\to X^*$
and $h\colon W^*\to Y^*$ such that $T=\{(g(w),h(w)) \mid w\in R\}$
(see~\cite{Berstel1979}). For a language $L\subseteq X^*$, we define
$TL=\{v\in Y^* \mid \exists u\in L\colon (u,v)\in T\}$.  A language
class $\C$ is a \emph{full trio} if for every language $L$ in $\C$,
the language $TL$ is effectively contained in $\C$ as well. Here,
``effectively'' means again that given a representation of a language
$L$ from $\C$ and a description of $T$, one can effectively compute a
representation of $TL$. For a language $L$, we denote by $\Trio{L}$
the smallest full trio containing $L$. Note that if $L\ne\emptyset$,
the class $\Trio{L}$ contains precisely the languages $TL$ for
rational transductions $T$.  For example, it is well-known that for
every monoid $M$, the class $\VA{M}$ is a full
trio~\cite{FernauStiebe2002a}.  A \emph{full AFL} is a full trio that
is also closed under Kleene iteration, i.e.  for each member $L$, the
language $L^*$ is effectively a member as well.

Here, we use the following fact. We denote the recursively enumerable
languages by $\RE$.
\begin{lem}\label{re:btwo}
Let $X=\{a_1,\bar{a}_1,b_1,a_2,\bar{a}_2,b_2\}$ and let $B_2\subseteq X^*$ be defined as
\[ B_2 = (\{a_1^n \bar{a}_1^n\mid n\ge 0\}b_1)^* \shuffle (\{a_2^n \bar{a}_2^n\mid n\ge 0\}b_2)^*. \]
Then $\RE$ equals $\Trio{B_2}$, the smallest full trio containing $B_2$.
\end{lem}

\Cref{re:btwo} is essentially due to Hartmanis and Hopcroft, who
stated it in slightly different terms:
\begin{thm}[Hartmanis and Hopcroft~\cite{HartmanisHopcroft1970}]\label{re:anbn}
Let $\C$ be the smallest full AFL containing $\{a^nb^n\mid n\ge 0\}$.  Every
recursively enumerable language is the homomorphic image of the intersection of
two languages in $\C$.
\end{thm}

By the following auxiliary result of Ginsburg and
Greibach~\cite[Theorem~3.2a]{GinsburgGreibach1970}, \cref{re:btwo} will follow
from \cref{re:anbn}.  
\begin{thm}[Ginsburg and Greibach~\cite{GinsburgGreibach1970}]\label{afl:trio}
Let $L\subseteq X^*$ and $c\notin X$. The smallest full AFL containing $L$
equals $\Trio{(Lc)^*}$.
\end{thm}

As announced, \cref{re:btwo} now follows. 
\begin{proof}[\cref{re:btwo}]\arev{I am not fully convinced of the interest of
  sharing the symbols $a_1$ and $a_2$ between $X$ and $X_\Gamma$, since the
    proof uses homomorphisms liberally.}
Since clearly $\Trio{B_2}\subseteq\RE$, it suffices to show
$\RE\subseteq\Trio{B_2}$.  According to \cref{re:anbn}, this amounts to
showing that $L_1\cap L_2\in\Trio{B_2}$ for any $L_1$ and $L_2$ in $\C$, where
$\C$ is the smallest full AFL containing the language $S=\{a^nb^n\mid n\ge 0\}$.  Hence, let
$L_1,L_2\in\C$. By \cref{afl:trio},  $L_1$ and $L_2$ belong to
$\C=\Trio{(Sc)^*}$.  This means we have
$L_i=T_i(\{a_i^n\bar{a}_i^n\mid n\ge 0\}b_i)^*$ for some rational transduction
$T_i$ for $i=1,2$.  Using a product construction, it is now easy to obtain a
rational transduction $T$ with $TB_2=L_1\cap L_2$.\myqed
\end{proof}

The proof of \cref{re:pfourcfour} will require one more auxiliary lemma.
In the following, $[w]_\Gamma$ denotes the congruence class of $w\in
X_\Gamma^*$ with respect to $\equiv_\Gamma$.
\begin{lem}\label{emptiness:graphmonoids:projection}
Let $\Gamma=(V,E)$ be a graph, let $W\subseteq V$ be a subset of vertices, and let $Y\subseteq
X_\Gamma$ be defined as $Y=\{a_w,\bar{a}_w \mid w\in W\}$. Then $u\equiv_\Gamma
v$ implies $\pi_Y(u)\equiv_\Gamma \pi_Y(v)$ for $u,v\in X_\Gamma^*$.
\end{lem}
\begin{proof}
An inspection of the rules in the Thue system $T_\Gamma$ reveals that if
$(u,v)\in R_\Gamma$, then either $(\pi_Y(u),\pi_Y(v))=(u,v)$ or
$\pi_Y(u)=\pi_Y(v)$. In any case, $\pi_Y(u)\equiv_\Gamma \pi_Y(v)$. Since
$\equiv_\Gamma$ is a congruence and $\pi_Y$ a morphism, this implies the
\lcnamecref{emptiness:graphmonoids:projection}.\myqed
\end{proof}
Note that the foregoing \lcnamecref{emptiness:graphmonoids:projection} does not
hold for arbitrary alphabets $Y\subseteq X_\Gamma$. For example, if $V=\{1\}$,
$X_\Gamma=\{a_1,\bar{a}_1\}$, and $Y=\{a_1\}$, then $a_1\bar{a}_1\equiv_\Gamma
\eword$, but $a_1\not\equiv_\Gamma\eword$.

We are now ready to prove \cref{re:pfourcfour}.
\begin{proof}[\cref{re:pfourcfour}]
  Observe that $w\equiv_\Gamma \eword$ if and only if $w$ can be
  transformed into $\eword$ by finitely many times replacing an infix
  $u$ with an infix $v$ for some $(u,v)\in R_\Gamma$.  Since
  $R_\Gamma$ is finite, this implies that the set of all $w\in
  X_\Gamma^*$ with $w\equiv_\Gamma\eword$ is recursively
  enumerable. (In fact, whether $w\equiv_\Gamma \eword$ can be decided
  in polynomial time~\cite{Zetzsche2013a,Zetzsche2016c}.) In
  particular, one can recursively enumerate runs of valence automata
  over $\VA{\M\Gamma}$ and hence $\VA{\M\Gamma}\subseteq\RE$. For the
  other inclusion, recall that $\VA{M}$ is a full trio for any monoid
  $M$. Furthermore, if $\Delta$ is an induced subgraph of $\Gamma$,
  then $\M\Delta$ embeds into $\M\Gamma$, meaning
  $\VA{\M\Delta}\subseteq\VA{\M\Gamma}$.  Hence, according to
  \cref{re:btwo}, it suffices to show that $B_2\in\VA{\M\Gamma}$ if
  $\Gamma^-$ equals $\Cfour$ or $\Pfour$.

\begin{figure}[t]
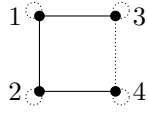

\begin{center}
\underlyingcfourpfour{1}{$1$}{$2$}{$3$}{$4$}
\end{center}
\caption{Graphs $\Gamma$ where $\Gamma^-$ is $\Cfour$ or $\Pfour$. Dotted lines
represent edges that may or may not exist in $\Gamma$.}
\label{membership:figure:pfourcfour}
\end{figure}
Let $X=\{a_1,\bar{a}_1,b_1,a_2,\bar{a}_2,b_2\}$.  and $\Gamma=(V,E)$. If
$\Gamma^-$ equals $\Cfour$ or $\Pfour$, then $V=\{1,2,3,4\}$ with
$\{3,1\},\{1,2\},\{2,4\}\in E$ and $\{1,4\},\{2,3\}\notin E$. See
\cref{membership:figure:pfourcfour}.  We construct a valence automaton $\cA$ over
$\M\Gamma$ for $B_2\subseteq X^*$ as follows.  First, $\cA$ reads a word in
$R=((a_1^*\bar{a}_1^*)b_1)^*\shuffle ((a_2^*\bar{a}_2^*)b_2)^*$. Here, when
reading $a_i$ or $\bar{a}_i$, it multiplies $[a_i]$ or $[\bar{a}_i]$,
respectively, to the storage monoid.  When reading $b_1$ or $b_2$, it
multiplies $[a_4]$ or $[a_3]$, respectively.  After this, $\cA$ switches to
another state and nondeterministically multiplies an element from $\{
[\bar{a}_4], [\bar{a}_3]\}^*$. Then it changes into an accepting state.  We
shall prove that $\cA$ accepts $B_2$. Let the morphism $h\colon X^*\to \{a_i,\bar{a}_i\mid
1\le i\le 4\}^*$ be defined by $h(a_i)=a_i$ and $h(\bar{a}_i)=\bar{a}_i$ for $i=1,2$ and
$h(b_1)=a_4$ and $h(b_2)=a_3$.

Suppose $w\in \Lang{\cA}$.  Then $w\in R$ and there is a $v\in
\{\bar{a}_4,\bar{a}_3\}^*$ with $[h(w)v]_\Gamma=[\eword]_\Gamma$.  Let
$w_i=\pi_{\{a_i,\bar{a}_i,b_i\}}(w)$.  If we can show $w_i\in
(\{a_i^n\bar{a}_i^n \mid n\ge 0\}^*b_i)^*$ for $i=1,2$, then clearly $w\in
B_2$.  For symmetry reasons, it suffices to prove this for $i=1$.  Let
$Y=\{a_1,\bar{a}_1,a_4,\bar{a}_4\}$.  Since
$[h(w)v]_\Gamma=[\eword]_\Gamma$, we have in particular
$[\pi_Y(h(w)v)]_\Gamma=[\eword]_\Gamma$ by
\cref{emptiness:graphmonoids:projection}. Moreover, 
\[ \pi_Y(h(w)v)=a_1^{n_1}\bar{a}_1^{\bar{n}_1}a_4\cdots a_1^{n_k}\bar{a}_1^{\bar{n}_k}a_4\bar{a}_4^m \]
for some $n_1,\ldots,n_k,\bar{n}_1,\ldots,\bar{n}_k,m\in\N$. Again, by
projecting to $\{a_4,\bar{a}_4\}^*$, we obtain
$[a_4^k\bar{a}_4^m]_\Gamma=[\eword]_\Gamma$ and hence $k=m$.  If
$n_k\ne\bar{n}_k$, then it is easy to see that $\pi_Y(h(w)v)$ cannot be reduced
to $\eword$, since there is no edge $\{1,4\}$ in $\Gamma$.  Therefore, we
have $n_k=\bar{n}_k$. It follows inductively that $n_i=\bar{n}_i$ for all $1\le
i\le k$.  Since $w_i=a_1^{n_1}\bar{a}_1^{\bar{n}_1}b_1\cdots
a_1^{n_k}\bar{a}_1^{\bar{n}_k}b_1$, this implies $w_i\in
(\{a_1^n\bar{a}_1^n\mid n\ge 0\}b_1)^*$.

We shall now prove $B_2\subseteq\Lang{\cA}$. Let $g\colon
X^*\to\{\bar{a}_3,\bar{a}_4\}$ be the morphism defined by $g(a_i)=g(\bar{a}_i)=\eword$
and $g(b_1)=\bar{a}_4$ and $g(b_2)=\bar{a}_3$.  We show by induction on $|w|$
that $w\in B_2$ implies $[h(w)\rev{g(w)}]_\Gamma=[\eword]_\Gamma$.  Since for
each $w\in B_2$, $\cA$ clearly has a run that puts $[h(w)\rev{g(w)}]_\Gamma$ into the
storage, this establishes $B_2\subseteq\Lang{\cA}$.  Suppose
$\pi_{\{b_1,b_2\}}(w)$ ends in $b_1$. Then $w=rsb_1$ for some $r\in X^*$, $s\in
(a_1^n\bar{a}_1^n)\shuffle t$ with $n\in\N$ and
$t\in\{a_2,\bar{a}_2,b_2\}^*$. Note that then $rt\in B_2$. Since there are
edges $\{1,2\},\{1,3\}$ in $\Gamma$, we have
$[h(s)]_\Gamma=[h(ta_1^n\bar{a}_1^n)]_\Gamma$. Moreover, since $g$ deletes
$a_1$ and $\bar{a}_1$, we have $g(s)=g(t)$. Therefore,
\begin{align*}
[h(w)\rev{g(w)}]_\Gamma &= [h(rsb_1)\rev{g(rsb_1)}]_\Gamma=[h(rta_1^n\bar{a}_1^n b_1)\rev{g(rtb_1)}]_\Gamma \\
 &= [h(rt)a_1^n\bar{a}_1^n a_4\bar{a}_4\rev{g(rt)}]_\Gamma=[h(rt)\rev{g(rt)}]_\Gamma.
\end{align*}
By induction, we have $[h(rt)\rev{g(rt)}]_\Gamma=[\eword]_\Gamma$ and hence
$[h(w)\rev{g(w)}]_\Gamma=[\eword]_\Gamma$.  If $\pi_{\{b_1,b_2\}}(w)$ ends in
$b_2$, then one can show $[h(w)\rev{g(w)}]_\Gamma=[\eword]_\Gamma$ completely
analogously.  This proves $B_2\subseteq\Lang{\cA}$ and hence the
\lcnamecref{re:pfourcfour}.\myqed
\end{proof}

\section{Decidability}\label{proofmain}
In this \lcnamecref{proofmain}, we prove
\cref{emptiness:graphmonoids:equivalence} and \cref{pnpgraphs}.
First, we mention existing results that are ingredients to our proofs.

Let $\C$ be a class of languages. A \emph{$\C$-grammar} is a quadruple
$G=(N,T,P,S)$ where $N$ and $T$ are disjoint alphabets and $S\in N$.
$P$ is a finite set of pairs $(A,M)$ with $A\in N$ and $M\subseteq
(N\cup T)^*$, $M\in\C$.  A pair $(A,M)\in P$ is called a
\emph{production of $G$}.  We write $x\grammarstep[G] y$ if $x=uAv$
and $y=uwv$ for some $u,v,w\in (N\cup T)^*$ and $(A,M)\in P$ with
$w\in M$. Moreover, $x \grammarstepsn[G]{n} y$ means that there are
$x_0,\ldots,x_n\in (N\cup T)^*$ with $x_{i-1} \grammarstep[G] x_i$ for $1\le i\le n$
and $x_0=x$ and $x_n=y$. Furthermore, we have $x\grammarsteps[G] y$ if $x\grammarstepsn[G]{n} y$
for some $n\ge 0$.
The \emph{language generated by $G$} is $\Lang{G}=\{ w\in
T^* \mid S\grammarsteps[G] w \}$. The class of all languages that are
generated by $\C$-grammars is called the \emph{algebraic extension of
  $\C$} and is denoted $\Alg{\C}$. Of course, if $\C\subseteq\D$, then
$\Alg{\C}\subseteq\Alg{\D}$. Moreover, it is easy to see that if
$\C\subseteq\Alg{\D}$, then $\Alg{\C\cup\D}=\Alg{\D}$.

The following is easy to show in the same way one shows that the
context-free languages constitute a full trio~\cite{Berstel1979}.
A proof can be found in \cite{Zetzsche2016c}.
\begin{lem}\label{alg-full-trio}
If $\C$ is a full trio, then $\Alg{\C}$ is a full trio as well.
\end{lem}

A monoid $M$ is called \emph{finitely generated} if there is a finite
subset $F\subseteq M$ such that every element of $M$ can
be written as a product of elements of $F$.  A language $I\subseteq
X^*$ is called an \emph{identity language for $M$} if there is a
surjective morphism $\varphi\colon X^*\to M$ with $I=\varphi^{-1}(1)$.
We will also use the following well-known fact about valence
automata. A proof can be found, e.g., in
\cite{Zetzsche2016c,Kambites2009}. 
\begin{prop}\label{va-smallest-full-trio}  
Let $M$ be a finitely generated monoid. Then:
\begin{enumerate}
\item $\VA{M}$ is the smallest full trio containing all identity
  languages of $M$.
\item If $L$ is any identity language of $M$, then $\VA{M}$ is the
  smallest full trio containing $L$.
\end{enumerate}
\end{prop}

The well-known theorem of Chomsky and
Sch\"{u}tzenberger~\cite{Berstel1979}, expressed in terms of valence
automata, states that $\VA{\Z*\Z}$ is the class of context-free
languages. This formulation, along with a new proof, is due to
Kambites~\cite{Kambites2009}.  Let $\Reg$ and $\CF$ denote the class
of regular and context-free languages, respectively. Then we have
$\Reg=\VA{\TrivialMonoid}$ and $\CF=\Alg{\Reg}$. Here,
$\TrivialMonoid$ denotes the trivial monoid $\{1\}$.  Moreover, notice
that $\Alg{\Alg{\C}}=\Alg{\C}$ for every language class $\C$. Since
furthermore valence automata over $\B*\B$ are equivalent to pushdown
automata, we have in summary:
\begin{equation} \CF=\VA{\B*\B}=\Alg{\VA{\TrivialMonoid}}=\Alg{\CF}=\VA{\Z*\Z}. \label{context-free}\end{equation}

In order to work with general free products, we use the following
result, which expresses the languages in $\VA{M_0*M_1}$ in terms of
$\VA{M_0}$ and $\VA{M_1}$. It was first shown in
\cite{Zetzsche2013a}. In \cite{BuckheisterZetzsche2013a}, it was
extended to more general products. For the convenience of the reader,
we include a proof.
\begin{prop}[\cite{Zetzsche2013a}]\label{freeproduct}
Let $M_0$ and $M_1$ be monoids. Then
$\VA{M_0*M_1}$ is included in $\Alg{\VA{M_0}\cup\VA{M_1}}$.
\end{prop}
\begin{proof}
  For every monoid $M$, we have $\VA{M}=\bigcup_N \VA{N}$, where $N$
  ranges over the finitely generated submonoids of $M$. Moreover,
  every finitely generated submonoid of $M_0*M_1$ is included in some
  $N_0*N_1$, where $N_i$ is a finitely generated submonoid of $M_i$,
  for $i=0,1$. Therefore, we have $\VA{M_0*M_1}=\bigcup_{N_0,N_1}
  \VA{N_0*N_1}$, where $N_i$ ranges over the finitely generated
  submonoids of $M_i$, for $i=0,1$. Thus, it suffices to show the
  \lcnamecref{freeproduct} in the case that $M_0$ and $M_1$ are
  finitely generated.

  For $i=0,1$, let $(A_i,R_i)$ be a presentation of $M_i$ such that
  $A_i$ is finite. Then $M_0*M_1$ is presented by $(A_0\cup
  A_1,R_0\cup R_1)$.  Consider the languages $L_i=\{w\in A_i^* \mid
  w\equiv_{R_i} \eword\}$ for $i\in\{0,1\}$. Then $L_i$ is an identity
  language of $M_i$ and hence contained in $\VA{M_i}$. Moreover, by
  definition of $M_0*M_1$, the language $L=\{w\in (A_0\cup A_1)^* \mid
  w\equiv_{R_0\cup R_1} \eword \}$ is an identity language of
  $M_0*M_1$.

  According to \cref{alg-full-trio}, the class
  $\Alg{\VA{M_0}\cup\VA{M_1}}$ is a full trio. Thus,
  \cref{va-smallest-full-trio} tells us that it suffices to show that
  the identity language $L$ of $M_0*M_1$ is contained in
  $\Alg{\VA{M_0}\cup\VA{M_1}}$.

  Consider the binary relation $\rightharpoonup$ on $(A_0\cup A_1)^*$
  where $u\rightharpoonup v$ if and only if for some $i\in\{0,1\}$,
  there are $x,z\in (A_0\cup A_1)^*$, $y\in A_i^*$, such that $u=xz$,
  $v=xyz$, and $y\equiv_{R_i} \eword$. It is now easy to see that
  $w\equiv_{R_0\cup R_1} \eword$ if and only if
  $\eword\rightharpoonup^* w$.

  This allows us to construct a $\VA{M_0}\cup\VA{M_1}$-grammar for
  $L$. Let $G=(N,A_0\cup A_1,P,S)$, where $N=\{S\}$. In order to
  describe the productions, we need to define two languages. For
  $i\in\{0,1\}$, let
  \[ K_i=\{Sa_1S\cdots a_nS\mid a_1,\ldots,a_n\in A_i,~~a_1\cdots
  a_n\in L_i\}. \] Then $K_i$ can be obtained from $L_i$ using full
  trio operations and is thus contained in $\VA{M_i}$. Our grammar
  contains only three productions: $S\to K_0$, $S\to K_1$, and $S\to
  \{\eword\}$ (recall that as a regular language, $\{\eword\}$ belongs
  to each $\VA{M_i}$). Then, it is immediate that $w\in \Lang{G}$ if
  and only if $\eword\rightharpoonup^* w$ and hence $\Lang{G}=L$.\myqed
\end{proof}

\Cref{freeproduct} tells us that the languages in $\VA{M_0 * M_1}$ are
confined to the algebraic extension of $\VA{M_0}\cup\VA{M_1}$.  Our
next ingredient, \cref{algbb}, will complement \cref{freeproduct} by
describing monoids $N$ such that the algebraic extension of $\VA{M}$
is confined to $\VA{N}$. We need two auxiliary lemmas, for which the
following notation will be convenient. We write $M\hookrightarrow N$
for monoids $M,N$ if there is a morphism $\varphi\colon M\to N$ such
that $\varphi^{-1}(1)=\{1\}$.  Clearly, if $M\hookrightarrow N$, then
$\VA{M}\subseteq\VA{N}$: Replacing in a valence automaton over $M$ all
elements $m\in M$ with $\varphi(m)$ yields a valence automaton over
$N$ that accepts the same language.
\begin{lem}\label{freeproduct-embedding}
If $M\hookrightarrow M'$ and $N\hookrightarrow N'$, then we have
$M*N\hookrightarrow M'*N'$.
\end{lem}
\begin{proof}
  Let $\varphi\colon M\to M'$ and $\psi\colon N\to N'$ be
  morphisms with $\varphi^{-1}(1)=\{1\}$ and
  $\psi^{-1}(1)=\{1\}$. Then defining $\kappa\colon M*N\to M'*N'$ as
  the morphism with $\kappa|_M=\varphi$ and $\kappa|_N=\psi$ clearly
  yields $\kappa^{-1}(1)=1$.\myqed
\end{proof}

For a monoid $M$, we define $\Rclass_1(M)=\{x\in M\mid \exists y\in
M\colon xy=1\}$.  Observe that the set $\Rclass_1(M)$ can be thought
of as the storage contents that can occur in a valid run of a valence
automaton over $M$. The following result appeared first
in~\cite{Zetzsche2015a}.  We include a proof for the convenience of
the reader.
\begin{lem}[\cite{Zetzsche2015a}]\label{bpowers}
  Let $M$ be a monoid with $\Rclass_1(M)\ne\{1\}$. Then we have 
  $\B^{(n)}*M\hookrightarrow \B*M$ for every $n\ge 1$.  In particular,
  $\VA{\B*M}=\VA{\B^{(n)}*M}$ for every $n\ge 1$.
\end{lem}
\begin{proof} 
  Observe that if $\B^{(n)}*M\hookrightarrow \B*M$ and
  $\B*\B*M\hookrightarrow \B*M$, then
  \[ \B^{(n+1)}*M\cong \B*(\B^{(n)}*M)\hookrightarrow
  \B*(\B*M)\hookrightarrow \B*M. \] Therefore, it suffices to prove
  $\B*\B*M\hookrightarrow \B*M$.

  Let $\B_s=\langle s,\bar{s}\mid s\bar{s}=1\rangle$ for
  $s\in\{p,q,r\}$. We show $\B_p*\B_q*M\hookrightarrow \B_r*M$.
  Suppose $M$ is presented by $(X,R)$.  We regard the monoids
  $\B_p*\B_q*M$ and $\B_r*M$ as embedded into $\B_p*\B_q*\B_r*M$,
  which by definition of the free product, has a presentation $(Y,S)$,
  where $Y=\{p,\bar{p},q,\bar{q},r,\bar{r}\}\cup X$ and $S$ consists
  of $R$ and the equations $s\bar{s}=1$ for $s\in \{p,q,r\}$.  For
  $w\in Y^*$, we write $[w]$ for the class of $w$ in the congruence
  generated by $S$.  Since $\Rclass_1(M)\ne\{1\}$, we find $u,v\in
  X^*$ with $[uv]=1$ and $[u]\ne 1$. 

  Observe that then for any $f,g\in (\{r,\bar{r}\}\cup X)^*$,
  we have $[frv\bar{r}g]\ne 1$: By induction on the number of rewriting steps,
  one can show that every word in $[frv\bar{r}g]$ is of the form $f'rv'\bar{r}g'$
  for $f',g'\in(\{r,\bar{r}\}\cup X)^*$ and $v'\in X^*$ with $v'\equiv_R v$.
  By the same argument, we have $[fru\bar{r}g]\ne 1$ for any $f,g\in (\{r,\bar{r}\}\cup X)^*$.

  Let $\varphi\colon
  (\{p,\bar{p},q,\bar{q}\}\cup X)^*\to(\{r,\bar{r}\}\cup X)^*$ be the
  morphism with $\varphi(x)=x$ for $x\in X$ and
  \begin{align*} p&\mapsto rr, & \bar{p}&\mapsto \bar{r}\bar{r}, \\
    q&\mapsto rur, & \bar{q}&\mapsto \bar{r}v\bar{r}.
  \end{align*} We show by induction on $|w|$ that $[\varphi(w)]=1$
  implies $[w]=1$. Since this is trivial for $w=\eword$, we assume
  $|w|\ge 1$.  Now suppose $[\varphi(w)]=[\eword]$ for some
  $w\in(\{p,\bar{p},q,\bar{q}\}\cup X)^*$. If $w\in X^*$, then
  $[\varphi(w)]=[w]$ and hence $[w]=1$.  Otherwise, we have
  $\varphi(w)=xry\bar{r}z$ for some $y\in X^*$ with $[y]=1$ and
  $[xz]=1$. This means $w=fsy\overline{s'}g$ for $s,s'\in\{p,q\}$ with
  $\varphi(fs)=xr$ and $\varphi(\overline{s'}g)=\bar{r}z$. If $s\ne
  s'$, then $s=p$ and $s'=q$; or $s=q$ and $s'=p$.  In the former case
  \[
  [\varphi(w)]=[\varphi(f)~rr~y~\bar{r}v\bar{r}~\varphi(g)]=[\varphi(f)rv\bar{r}\varphi(g)]\ne
  1 \] by our observation above and in the latter
  \[
  [\varphi(w)]=[\varphi(f)~rur~y~\bar{r}\bar{r}~\varphi(g)]=[\varphi(f)ru\bar{r}\varphi(g)]\ne
  1, \] again by our observation. Hence $s=s'$. This means
  $[w]=[fsy\bar{s}g]=[fg]$ and also
  $1=[\varphi(w)]=[\varphi(fg)]$ and since $|fg|<|w|$,
  induction yields $[w]=[fg]=1$.

  Hence, we have shown that $[\varphi(w)]=1$ implies $[w]=1$. Since,
  on the other hand, $[u]=[v]$ implies $[\varphi(u)]=[\varphi(v)]$ for
  all $u,v\in (\{p,\bar{p},q,\bar{q}\}\cup X)^*$, we can lift
  $\varphi$ to a morphism witnessing $\B_p*\B_q*M\hookrightarrow
  \B_r*M$.\myqed
\end{proof}

As a partial converse to \cref{freeproduct}, we have the following.
It was first shown in~\cite{Zetzsche2015a}.  Since valence
automata over $\B*\B$ are essentially pushdown automata and since
$\Alg{\VA{\TrivialMonoid}}=\Alg{\Reg}=\CF$, the equality
$\VA{\B*\B*M}=\Alg{\VA{M}}$ generalizes the equivalence between
pushdown automata and context-free grammars. 
\begin{prop}[\cite{Zetzsche2015a}]\label{algbb}
For every monoid $M$, $\VA{\B*\B*M}=\Alg{\VA{M}}$. Moreover, if
$\Rclass_1(M)\ne\{1\}$, then $\VA{\B*M}=\Alg{\VA{M}}$.
\end{prop}
\begin{proof}
  It suffices to prove the first statement: If $\Rclass_1(M)\ne\{1\}$,
  then \cref{bpowers} implies $\VA{\B*M}=\VA{\B*\B*M}$.  Observe that
  if $\C$ is a language class with $\C\subseteq\CF$, then $\Alg{\C\cup
    \VA{N}}=\Alg{\VA{N}}$ for every monoid $N$: For each production $A\to L$ in a
  $(\C\cup\VA{N})$-grammar $G$ with $L$ from $\C$, we can take a
  context-free grammar $G'$ generating $L$ (with fresh non-terminals) and replace the production
  $A\to L$ with the productions of $G'$. This is yields a $\VA{N}$-grammar because all
  singleton sets are contained in $\VA{N}$. Therefore, 
since the languages in $\VA{\B}$ are effectively context-free, \cref{freeproduct} yields
\[ \VA{\B*N}\subseteq\Alg{\VA{\B}\cup\VA{N}}=\Alg{\VA{N}} \]
for every monoid $N$. Therefore,
\[ \VA{\B*\B*M}\subseteq\Alg{\VA{\B*M}}\subseteq\Alg{\Alg{\VA{M}}}=\Alg{\VA{M}}. \]
It remains to be shown that $\Alg{\VA{M}}\subseteq\VA{\B*\B*M}$.

Suppose $G=(N,T,P,S)$ is a $\VA{M}$-grammar and let $X=N\cup T$. Since
$\VA{M}$ is closed under union, we may assume that for each $B\in N$, there is
precisely one production $B\to L_B$ in $P$.  For each nonterminal $B\in N$,
there is a valence automaton $\cA_B=(Q_B,X,M,E_B,q^B_0,F_B)$ over $M$ with
$\Lang{\cA_B}=L_B$.  We may clearly assume that $Q_B\cap Q_C=\emptyset$ for $B\ne
C$ and that for each $(p,w,m,q)\in E_B$, we have $|w|\le 1$.

\newcommand{\bleft}{\lfloor}
\newcommand{\bright}{\rfloor}
In order to simplify the correctness proof, we modify $G$.  Let $\bleft$ and
$\bright$ be new symbols and let $G'$ be the grammar
$G'=(N,T\cup\{\bleft,\bright\},P',S)$, where $P'$ consists of the productions
$B\to \bleft L\bright$ for $B\to L\in P$.
Moreover, let
\[ K=\{v\in (N\cup T\cup\{\bleft,\bright\})^* \mid u \grammarsteps[G'] v,~u\in L_S\}. \]
Then $\Lang{G}=\pi_T(K\cap (T\cup\{\bleft,\bright\})^*)$ and it suffices to
show $K\in\VA{\B*\B*M}$.

Let $Q=\bigcup_{B\in N} Q_B$. For each $q\in Q$, let $\B_q=\langle
q,\bar{q}\mid q\bar{q}=1\rangle$ be an isomorphic copy of $\B$.  Let
$M'=\B_{q_1}*\cdots*\B_{q_n}*M$, where $Q=\{q_1,\ldots,q_n\}$.  We shall prove
$K\in\VA{M'}$, which implies $K\in\VA{\B*\B*M}$ by \cref{bpowers} since $\Rclass_1(\B*M)\ne\{1\}$.

Let $E=\bigcup_{B\in N} E_B$, $F=\bigcup_{B\in N} F_B$.  The new set $E'$
consists of the following transitions:
\begin{align}
&(p,x,m,q)              && \text{for $(p,x,m,q)\in E$,} \label{basics:algbb:t:old} \\
&(p,\bleft,mq,q^B_0)    && \text{for $(p,B,m,q)\in E$, $B\in N$,} \label{basics:algbb:t:open} \\
&(p,\bright,\bar{q},q)  && \text{for $p\in F$, $q\in Q$.} \label{basics:algbb:t:close}
\end{align}
We claim that with $\cA'=(Q, N\cup T\cup\{\bleft,\bright\},M',E',q_0^S,F)$, we have
$\Lang{\cA'}=K$.

Let $v\in K$, where $u\grammarstepsn[G']{n} v$ for some $u\in L_S$. We show
$v\in\Lang{\cA'}$ by induction on $n$. For $n=0$, we have $v\in L_S$ and can use
transitions of type \labelcref{basics:algbb:t:old} inherited from $\cA_S$ to
accept $v$. If $n\ge 1$, let $u\grammarstepsn[G']{n-1} v'\grammarstep[G'] v$.
Then $v'\in \Lang{\cA'}$ and $v'=xBy$, $v=x\bleft w\bright y$ for some $B\in N$,
$w\in L_B$. The run for $v'$ uses a transition $(p,B,m,q)\in E$. Instead of
using this transition, we can use $(p,\bleft,mq,q_0^B)$, then execute the
\labelcref{basics:algbb:t:old}-type transitions for $w\in L_B$, and finally use
$(f,\bright,\bar{q},q)$, where $f$ is the final state in the run for $w$.  This
has the effect of reading $\bleft w\bright$ from the input and multiplying
$mq1\bar{q}=m$ to the storage monoid.  Hence, the new run is valid and accepts
$v$. Hence, $v\in \Lang{\cA'}$. This proves $K\subseteq\Lang{\cA'}$.

In order to show $\Lang{\cA'}\subseteq K$, consider the morphisms
$\varphi\colon (T\cup\{\bleft,\bright\})^*\to \B$, $\psi\colon M'\to\B$ with
$\varphi(x)=1$ for $x\in T$, $\varphi(\bleft)=a$, $\varphi(\bright)=\bar{a}$,
$\psi(q)=a$ for $q\in Q$, $\psi(\bar{q})=\bar{a}$, and $\psi(m)=1$ for $m\in
M$.  The transitions of $\cA'$ are constructed such that
$(p,\eword,1)\autsteps[\cA'] (q,w,m)$ implies $\varphi(w)=\psi(m)$. In
particular, if $v\in \Lang{\cA'}$, then $\pi_{\{\bleft,\bright\}}(v)$ is a semi-Dyck
word with respect to $\bleft$ and $\bright$.

Let $v\in\Lang{\cA'}$ and let $n=|w|_{\bleft}$. We show $v\in K$ by induction on
$n$.  If $n=0$, then the run for $v$ only used transitions of type
\labelcref{basics:algbb:t:old} and hence $v\in L_S$. If $n\ge 1$, since
$\pi_{\{\bleft,\bright\}}(v)$ is a semi-Dyck word, we can write $v=x\bleft
w\bright y$ for some $w\in (N\cup T)^*$. Since $\bleft$ and $\bright$ can only
be produced by transitions of the form \labelcref{basics:algbb:t:open} and
\labelcref{basics:algbb:t:close}, respectively, the run for $v$ has to be of
the form
\begin{align*}
 (q_0^S,\eword,1)&\autsteps[\cA'](p,x,r) \\
                     &\autstep[\cA'](q_0^B,x\bleft,rmq) \\
                     &\autsteps[\cA'](f,x\bleft w,rmqs) \\
                     &\autstep[\cA'](q',x\bleft w\bright,rmqs\overline{q'}) \\
                     &\autsteps[\cA'] (f',x\bleft w\bright y,rmqs\overline{q'}t)
\end{align*}
for some $p,q,q'\in Q$, $B\in N$, $(p,B,m,q)\in E$, $f,f'\in F$, $r,t\in M'$,
and $s\in M$ and with $rmqs\overline{q'}t=1$. This last condition implies $s=1$
and $q=q'$, which in turn entails $rmt=1$. This also means
$(p,B,m,q')=(p,B,m,q)\in E$ and $(q_0^B,\eword,1)\autsteps[\cA']
(f,w,s)=(f,w,1)$ and hence $w\in L_B$.  Using the transition $(p,B,m,q')\in E$, we
have
\begin{align*}
(q_0^S,\eword,1) &\autsteps[\cA'] (p,x,r) \\
                     &\autstep[\cA'] (q',xB,rm) \\
                     &\autsteps[\cA'] (f',xBy,rmt).
\end{align*}
Hence $xBy\in\Lang{\cA'}$ and $|xBy|_{\bleft}<|v|_{\bleft}$. Thus,
induction yields $xBy\in K$ and since $xBy\grammarstep[G'] x\bleft
w\bright y$, we have $v=x\bleft w\bright y\in K$.  This proves
$\Lang{\cA'}=K$.\myqed
\end{proof}

\newcommand{\triointersect}{\sqcap} 

For two language classes $\C$ and $\D$, we will consider the languages
obtained by intersecting a language from $\C$ with a language in
$\D$. Since the class of these intersections might not be
well-behaved, we use a slight extension.  By $\C\triointersect \D$, we
denote the class of all languages $h(K\cap L)$ where $K\subseteq X^*$
belongs to $\C$ and $L\subseteq X^*$ is a member of $\D$ and $h\colon
X^*\to Y^*$ is a morphism. This allows us to state the following
characterization of $\VA{M\times N}$ in terms of $\VA{M}$ and $\VA{N}$
by Kambites~\cite{Kambites2009}.
\begin{prop}\label{directproduct}
If $M,N$ are monoids, then $\VA{M\times N}=\VA{M}\triointersect \VA{N}$.
\end{prop}
This implies in particular that if $\VA{M_i}\subseteq\VA{N_i}$ for
$i\in\{0,1\}$, then we also have the inclusion $\VA{M_0\times
  M_1}\subseteq \VA{N_0\times N_1}$.  Of course, this also means that
if $\VA{M_i}=\VA{N_i}$ for $i\in\{0,1\}$, then $\VA{M_0\times M_1}=\VA{N_0\times N_1}$.
We are now ready to prove
\cref{pnpgraphs}.

\begin{proof}[\cref{pnpgraphs}]\label{proof-pnpgraphs}
  By definition, we have $\M\Gamma\cong \B\times (M_0*M_1)$, where
  $M_i\cong\B$ or $M_i\cong \Z$ for $i\in\{0,1\}$. We show that
  $\VA{M_0*M_1}=\VA{\B*\B}$ in any case. This suffices, since it
  clearly implies $\VA{\M\Gamma}=\VA{\B^{(2)}\times\B}$ according to
  \cref{directproduct}. If $M_0\cong M_1\cong\B$, the equality
  $\VA{M_0*M_1}=\VA{\B*\B}$ is trivial, so we may assume $M_0\cong\Z$.

  If $M_1\cong\Z$, then $M_0*M_1\cong\Z*\Z$, meaning that
  $\VA{M_0*M_1}$ is the class of context-free languages (see
  \cref{context-free}) and thus $\VA{M_0*M_1}=\VA{\B*\B}$.

  If $M_1\cong\B$, then $\VA{\Z*\B}=\Alg{\VA{\Z}}$ by
  \cref{algbb}. Since $\VA{\Z}$ is included in the context-free
  languages, we have $\Alg{\VA{\Z}}=\VA{\B*\B}$.\myqed
\end{proof}

We shall now prove \cref{emptiness:graphmonoids:equivalence}. Note
that the implication \DesImp{dec:cond:decidable}{dec:cond:cfourpfour}
immediately follows from \cref{re:pfourcfour}. The implication
\DesImp{dec:cond:cfourpfour}{dec:cond:transitive} is an old
graph-theoretic result of Wolk.
\begin{thm}[Wolk~\cite{Wolk1965}]\label{transitiveforest}
A simple graph is a transitive forest if and only if it does not
contain $\Cfour$ or $\Pfour$ as an induced subgraph.
\end{thm}

The implication \DesImp{dec:cond:transitive}{dec:cond:monclass} is a simple
combinatorial observation. An analogous fact is part of Lohrey and Steinberg's
proof of \cref{ratmp:graphgroups}.
\begin{lem}\label{emptiness:forest:monclass}
If $\Gamma$ is a \PNP-free transitive forest, then $\M\Gamma\in\MonDecidable$.
\end{lem}
\begin{proof}
Let $\Gamma=(V,E)$. We proceed by induction on $|V|$. 
If $\Gamma$ is empty, then
$\M\Gamma\cong\TrivialMonoid\cong\B^0\in\MonDecidable$.  Hence, we assume that
$\Gamma$ is non-empty.  If $\Gamma$ is not connected, then $\Gamma$ is the
disjoint union of \PNP-free transitive forests $\Gamma_1,\Gamma_2$, for which
$\M\Gamma_1,\M\Gamma_2\in\MonDecidable$ by induction. Hence,
$\M\Gamma\cong\M\Gamma_1*\M\Gamma_2\in\MonDecidable$. 

Suppose $\Gamma$ is connected. Since $\Gamma$ is a transitive forest, there is
a vertex $v\in V$ such that $\Gamma\setminus v$ is a \PNP-free transitive
forest and $v$ is adjacent to every vertex in $V\setminus\{v\}$. We distinguish
two cases.
\begin{proofcases}
\item If $v$ is a looped vertex, then
$\M\Gamma\cong\Z\times\M(\Gamma\setminus v)$, and
$\M(\Gamma\setminus v)\in\MonDecidable$ by induction.
\item If $v$ is an unlooped vertex, then $\Gamma$ being \PNP-free means that in
$\Gamma\setminus v$, any two distinct vertices are adjacent. Hence,
$\M\Gamma\cong \B^m\times\Z^n$, where $m$ and $n$ are the number of unlooped
and looped vertices in $\Gamma$, respectively. Therefore,
$\M\Gamma\in\MonDecidable$.\myqed
\end{proofcases}
\end{proof}

For establishing \cref{emptiness:graphmonoids:equivalence}, our
remaining task is to prove the implication
\DesImp{dec:cond:monclass}{dec:cond:decidable}.  In light of
\cref{re:pfourcfour,transitiveforest,emptiness:forest:monclass}, this
amounts to showing that emptiness is decidable for valence automata
over monoids in $\MonDecidable$. This will involve two facts
(\cref{algparikh} and \cref{productz}) about the languages arising
from monoids in $\MonDecidable$.

The following generalization of Parikh's theorem by van~Leeuwen will
allow us to exploit our description of free products by algebraic
extensions.  If $X$ is an alphabet, $X^\oplus$ denotes the set of maps
$\alpha\colon X\to\N$. The elements of $X^\oplus$ are called
\emph{multisets}.  The \emph{Parikh map} is the map $\ParikhMap\colon
X^*\to X^\oplus$ where $\Parikh{w}(x)$ is the number of occurrences of
$x$ in $w$.  By $\Powerset{S}$, we denote the power set of the set
$S$. A \emph{substitution} is a map $\sigma\colon X\to \Powerset{Y^*}$,
where $X$ and $Y$ are alphabets.
Given $L\subseteq X^*$, we write $\sigma(L)$ for the set of all
words $v_1\cdots v_n$, where $v_i\in \sigma(x_i)$, $1\le i\le n$, for
$x_1\cdots x_n\in L$ and $x_1,\ldots,x_n\in X$. If $\sigma(x)$ belongs
to $\C$ for each $x\in X$, then $\sigma$ is a
\emph{$\C$-substitution}. The class $\C$ is said to be
\emph{substitution closed} if $\sigma(L)\in\C$ for every member $L$ of
$\C$ and every $\C$-substitution $\sigma$. 
\begin{thm}[van~Leeuwen~\cite{vanLeeuwen1974}]\label{algparikh}
For each substitution closed full trio $\C$, we have
$\Parikh{\Alg{\C}}=\Parikh{\C}$.
\end{thm}

For $\alpha,\beta\in X^\oplus$, let $\alpha+\beta\in X^\oplus$ be defined by
$(\alpha+\beta)(x)=\alpha(x)+\beta(x)$. With this operation, $X^\oplus$ is a
monoid. For a subset $S\subseteq X^\oplus$, we write $S^\oplus$ for the
smallest submonoid of $X^\oplus$ containing $S$. A subset of the form
$\mu+F^\oplus$ for $\mu\in X^\oplus$ and a finite $F\subseteq X^\oplus$ is
called \emph{linear}. A finite union of linear sets is called
\emph{semilinear}. By $\HomSLI{\C}$ we denote the class of languages
$h(L\cap\ParikhInv{S})$, where $h\colon X^*\to Y^*$ is a morphism, $L$ belongs
to $\C$, and $S\subseteq X^\oplus$ is semilinear. 
\begin{prop}[\cite{Zetzsche2015a}]\label{productz}
For each monoid $M$, we have 
\[ \HomSLI{\VA{M}}=\bigcup_{i\ge 0}\VA{M\times\Z^i}. \]
\end{prop}

We will prove decidability for $\MonDecidable$ by
reducing the problem to the reachability problem of priority
multicounter machines, whose decidability has been established by
Reinhardt~\cite{Reinhardt2008}. Priority multicounter
machines are an extension of Petri nets with one inhibitor arc.
Intuitively, a priority multicounter machine is a partially blind
multicounter machine with the additional capability of restricted
zero tests: The counters are numbered from $1$ to $k$ and for each
$\ell\in\{1,\ldots,k\}$, there is a zero test instruction that checks
whether counters $1$ through $\ell$ are zero. Let us define priority
multicounter machines formally.

\label{def:prio}
A \emph{priority $k$-counter machine} is a tuple $\cA=(Q, X, E, q_0, F)$, where
\begin{enumerate*}[label=(\roman*)]
\item $X$ is an alphabet,
\item $Q$ is a finite set of states,
\item $E$ is a finite subset of $Q\times X^*\times \{0,\ldots,k\}\times \Z^k\times Q$,
and its elements are called \emph{edges} or \emph{transitions},
\item $q_0\in Q$ is the \emph{initial state}, and
\item $F\subseteq Q$ is the set of \emph{final states}.
\end{enumerate*}
For $\ell\in\{0,\ldots,k\}$, let 
\[\N^k_\ell=\{(\mu_1,\ldots,\mu_k)\in\N^k \mid \mu_1=\ldots=\mu_\ell=0 \}. \]
We are now ready to defines the semantics of priority counter
machines. A \emph{configuration} of $\cA$ is a pair
$(q,\mu)\in Q\times\N^k$. For configurations $(q, \mu)$ and $(q',
\mu')$, we write $(q,\mu)\lautsteps[\cA]{w} (q',\mu')$ if there are
$(q_0,\mu_0),\ldots,(q_n,\mu_n)\in Q\times\N^k$ such that
\begin{enumerate}[label=(\roman*)]
\item $(q,\mu)=(q_0,\mu_0)$ and $(q',\mu')=(q_n,\mu_n)$,
\item for each $i\in\{1,\ldots,n\}$, there is a transition
  $(q_{i-1},w_i,\ell,\nu,q_i)\in E$ such that $\mu_{i-1}\in \N^k_\ell$
  and $\mu_i=\mu_{i-1}+\nu$, and $w=w_1\cdots w_n$.
\end{enumerate}
The \emph{language accepted by $\cA$} is defined as
\[ \Lang{\cA} = \{w\in X^* \mid (q_0, 0) \lautsteps[\cA]{w} (f, 0)~~\text{for some $f\in F$}\}. \]
A \emph{priority multicounter machine} is a priority $k$-counter machine for
some $k\in\N$.  The class of languages accepted by priority multicounter
machines is denoted by $\Prio$.
Reinhardt has shown that the reachability problem for
priority multicounter machines is decidable~\cite{Reinhardt2008}, which
can be reformulated as follows.
\begin{thm}[Reinhardt~\cite{Reinhardt2008}]\label{prio:decidable}
The emptiness problem is decidable for priority multicounter machines.
\end{thm}

Although the decidability proof for the emptiness problem for valence
automata over $\MonDecidableSimple$ employs a reduction to priority
multicounter machines, it should be stressed that the mechanisms
realized by $\MonDecidableSimple$ are quite different from priority
counters and very likely not subsumed by them in terms of accepted
languages. For example, $\MonDecidableSimple$ contains pushdown stacks
($\B*\B$)---if the priority multicounter machines could accept all
context-free languages (or even just the semi-Dyck language $D_2^*$),
this would easily imply decidability of the emptiness problem for
pushdown Petri nets. Indeed, $\MonDecidableSimple$ can even realize
stacks where each entry consists of $n$ partially blind counters
(since $\B*(\B^n)\in\MonDecidableSimple$). On the other hand, priority
multicounter machines do not seem to be subsumed by
$\MonDecidableSimple$ either: After building stacks once,
$\MonDecidableSimple$ only allows adding \emph{blind} counters (and
building stacks again). It therefore seems unlikely that a mechanism
in $\MonDecidableSimple$ can accept the languages even of a priority
$2$-counter machine.

The idea of the proof of \DesImp{dec:cond:monclass}{dec:cond:decidable}
is, given a valence automaton over some $M\in\MonDecidable$, to
construct a Parikh-equivalent priority multicounter machine. This
construction makes use of the following simple fact. A full trio $\C$ is
said to be \emph{Presburger closed} if $\HomSLI{\C}\subseteq\C$.
\begin{lem}\label{prio:closure}
$\Prio$ is a Presburger closed full trio and closed under substitutions.
\end{lem}
\begin{proof}
The fact that $\Prio$ is a full trio can be shown by standard automata
constructions.  Given a priority multicounter machine $\cA$ and a semilinear set
$S\subseteq X^\oplus$, we add $|X|$ counters to $\cA$ that ensure that the input
is contained in $\Lang{\cA}\cap \ParikhInv{S}$. This proves that $\Prio$ is
Presburger closed.

Suppose $\sigma\colon X\to\Powerset{Y^*}$ is a $\Prio$-substitution.
Furthermore, let $\cA$ be a priority $k$-counter machine and let $\sigma(x)$ be
given by a priority $\ell$-counter machine for each $x\in X$. We construct a
priority $(\ell+k)$-counter machine $\cB$ from $\cA$ by adding $\ell$ counters. $\cB$
simulates $\cA$ on counters $\ell+1,\ldots,\ell+k$. Whenever $\cA$ reads $x$, $\cB$
uses the first $\ell$ counters to simulate the priority $\ell$-counter machine
for $\sigma(x)$. Using the zero test on the first $\ell$ counters, it makes
sure that the machine for $\sigma(x)$ indeed ends up in a final configuration.
Then clearly $\Lang{\cB}=\sigma(\Lang{\cA})$.\myqed
\end{proof}

\begin{lem}\label{emptiness:parikh}
We have the effective inclusion
$\Parikh{\VA{\MonDecidable}}\subseteq\Parikh{\Prio}$. More precisely, given
$M\in\MonDecidable$ and $L\in\VA{M}$, one can construct an $L'\in\Prio$ with
$\Parikh{L'}=\Parikh{L}$.
\end{lem}
\begin{proof}
We proceed by induction with respect to the definition of $\MonDecidable$. In
the case $M=\B^n$, we have $\VA{M}\subseteq\Prio$, because priority
multicounter machines generalize partially blind multicounter machines.

Suppose $M=N\times\Z$ and $\Parikh{\VA{N}}\subseteq\Parikh{\Prio}$ and let
$L\in\VA{M}$.  By \cref{productz}, we have $L=h(K\cap\ParikhInv{S})$
for some semilinear set $S$, a morphism $h$, and $K\in\VA{N}$. Hence, there is
a $\bar{K}\in\Prio$ with $\Parikh{\bar{K}}=\Parikh{K}$. With this, we have
$\Parikh{L}=\Parikh{h(\bar{K}\cap\ParikhInv{S})}$ and since $\Prio$ is
Presburger closed, we have $h(\bar{K}\cap\ParikhInv{S})\in\Prio$ and thus
$\Parikh{L}\in\Parikh{\Prio}$.

Suppose $M=M_0*M_1$ and $\Parikh{\VA{M_i}}\subseteq\Parikh{\Prio}$ for
$i\in\{0,1\}$. Let $L$ be a member of $\VA{M}$. According to
\cref{freeproduct}, this means $L$ belongs to $\Alg{\VA{M_0}\cup\VA{M_1}}$. Since
$\Parikh{\VA{M_0}\cup\VA{M_1}}\subseteq\Parikh{\Prio}$, we can  
construct a $\Prio$-grammar $G$ with $\Parikh{\Lang{G}}=\Parikh{L}$.  By
\cref{algparikh} and \cref{prio:closure}, this implies
$\Parikh{L}\in\Parikh{\Prio}$.\myqed
\end{proof}

The following \lcnamecref{emptiness:decidable} is a direct consequence of
\cref{emptiness:parikh} and \cref{prio:decidable}: Given a valence automaton
over $M$ with $M\in\MonDecidable$, we construct a priority multicounter machine
accepting a Parikh-equivalent language. The latter can then be checked for
emptiness.
\begin{lem}\label{emptiness:decidable}
For each $M\in\MonDecidable$, the emptiness problem for valence automata over
$M$ is decidable.
\end{lem}
This completes the proof of
\DesImp{dec:cond:monclass}{dec:cond:decidable} of
\cref{emptiness:graphmonoids:equivalence} and hence concludes the
proof of \cref{emptiness:graphmonoids:equivalence}.

\section{Expressive equivalences}\label{proofexpressive}

We now turn to the proof of
\cref{emptiness:decidable:intuition,emptiness:remaining:intuition}, which
characterize the expressiveness of valence automata over $\MonDecidableSimple$
and $\MonRemaining$, respectively.

\begin{proof}[\cref{emptiness:decidable:intuition}] Since
$\MonDecidableSimple\subseteq\MonDecidable$, the inclusion
``$\supseteq$'' is immediate. We show by induction with respect to
the definition of $\MonDecidable$ that for each $M\in\MonDecidable$,
there is an $M'\in\MonDecidableSimple$ with $\VA{M}\subseteq\VA{M'}$.
This is trivial if $M=\B^n$, so suppose $\VA{M}\subseteq\VA{M'}$
and $\VA{N}\subseteq\VA{N'}$ for $M,N\in\MonDecidable$ and
$M',N'\in\MonDecidableSimple$. Observe that by induction on the
definition of $\MonDecidableSimple$, one can show that there is
a common $P\in\MonDecidableSimple$ with $\VA{M'}\subseteq\VA{P}$
and $\VA{N'}\subseteq\VA{P}$. Of course, we may assume that
$\Rclass_1(P)\ne\{1\}$. Then we have
\begin{align*}
\VA{M*N}&\subseteq\Alg{\VA{M}\cup\VA{N}}\subseteq\Alg{\VA{M'}\cup\VA{N'}} \\
        &\subseteq\Alg{\VA{P}}= \VA{\B*P}, 
\end{align*}
in which the first inclusion is due to \cref{freeproduct}
and the equality in the end is provided by \cref{algbb}.
Since $\B*P\in\MonDecidableSimple$, this completes the
proof for $M*N$. Moreover, $\VA{M}\subseteq\VA{M'}$ implies
$\VA{M\times\Z}\subseteq\VA{M'\times\Z}$ and we have
$M'\times\Z\in\MonDecidableSimple$.\myqed
\end{proof}

\begin{proof}[\cref{emptiness:remaining:intuition}]
By induction, it is easy to see that each $M\in\MonRemainingSimple$ is
isomorphic to some $\M\Gamma$ where $\Gamma$ contains a \PNP-graph and
$\Gamma^-$ is a transitive forest.  By \cref{transitiveforest}, this means
$\Gamma^-$ contains neither $\Cfour{}$ nor $\Pfour{}$. This proves the
inclusion ``$\supseteq$''. 

Because of \cref{transitiveforest}, for the inclusion ``$\subseteq$'', it
suffices to show that if $\Gamma^-$ is a transitive forest, then there is some
$M\in\MonRemainingSimple$ with $\VA{\M\Gamma}\subseteq\VA{M}$. We prove this by
induction on the number of vertices in $\Gamma=(V,E)$. As in the proof of
\cref{emptiness:forest:monclass}, we may assume that for every induced proper
subgraph $\Delta$ of $\Gamma$, we find an $M\in\MonRemainingSimple$ with
$\VA{\M\Gamma}\subseteq\VA{M}$.
If $\Gamma$ is empty, then $\M\Gamma\cong\TrivialMonoid$ and
$\VA{\M\Gamma}\subseteq\VA{\B^{(2)}\times\B}$. Hence, we may assume that
$\Gamma$ is non-empty.

If $\Gamma$ is not connected, then $\Gamma=\Gamma_1\uplus\Gamma_2$ with non-empty graphs
$\Gamma_1,\Gamma_2$. This implies that there are $M_1,M_2\in\MonRemainingSimple$ with
$\VA{\M\Gamma_i}\subseteq\VA{M_i}$ for $i\in\{1,2\}$. By induction with respect
to the definition of $\MonRemainingSimple$, one can show that there is a common
$N\in\MonRemainingSimple$ with $\VA{M_i}\subseteq\VA{N}$ for $i\in\{1,2\}$.
Here, $N$ can clearly be chosen with $\Rclass_1(N)\ne\{1\}$. Then, we have
\begin{align*}
\VA{\M\Gamma}&=\VA{\M\Gamma_1*\M\Gamma_2}\subseteq\Alg{\VA{\M\Gamma_1}\cup\VA{\M\Gamma_2}} \\
&\subseteq\Alg{\VA{M_1}\cup\VA{M_2}}\subseteq\Alg{\VA{N}}= \VA{\B*N}
\end{align*}
and $\B*N\in\MonRemainingSimple$ as in the proof of
\cref{emptiness:decidable:intuition}.

Suppose $\Gamma$ is connected. Since $\Gamma^-$ is a transitive
forest, there is a vertex $v\in V$ that is adjacent to every vertex in
$V\setminus\{v\}$. By induction, there is an $M\in\MonRemainingSimple$
with $\VA{\M(\Gamma\setminus v)}\subseteq\VA{M}$. Depending on whether
$v$ is looped or not, we have $\M\Gamma\cong\M(\Gamma\setminus
v)\times\Z$ or $\M\Gamma\cong\M(\Gamma\setminus v)\times\B$. Since
$\VA{\Z}\subseteq\VA{\B\times\B}$ (one blind counter can easily be
simulated by two partially blind counters), this yields
$\VA{\M\Gamma}\subseteq\VA{\M(\Gamma\setminus v)\times\B\times\B} \subseteq \VA{M\times\B\times\B}$
and the fact that $M\times\B\times\B\in\MonRemainingSimple$ completes the proof.\myqed
\end{proof}

\section{Finite-index languages and Petri nets}\label{synthesis}
We have seen in the previous sections that valence automata over
$\MonDecidableSimple$ constitute a model that (strictly) subsumes Petri nets
and has a decidable emptiness problem. Moreover, they feature a type of
pushdown stack. A similar result has been obtained by Atig and
Ganty~\cite{AtigGanty2011}. They proved that given a finite-index context-free
language and a Petri net language, it is decidable whether their intersection
is empty. Here, we present a common generalization of these facts: We provide a
language class that contains the languages considered by Atig and Ganty and
those in $\VA{\MonDecidableSimple}$ and enjoys decidability of the emptiness
problem.

\paragraph{Definitions} If $\C$ is the class of finite languages, a
$\C$-grammar is also called a \emph{context-free grammar}. For a context-free
grammar $G=(N,T,P,S)$, we may assume that for each production $(A,M)$, the set
$M$ is a singleton and instead of $(A,\{w\})$, we write $A\to w$ for the
production. The grammar is said to be in \emph{Chomsky normal form (CNF)} if 
for every production $A\to w$, we have $w\in N^2\cup T\cup\{\eword\}$.

The finite-index restriction considered by Atig and Ganty places a
budget constraint on the nonterminal occurrences in sentential forms. This
leads to a restricted derivation relation. Suppose $G$ is in CNF. For $u,v\in(N\cup T)^*$, we write
$u\grammarstep[G,k] v$ if $|u|_N,|v|_N\le k$ and $u\grammarstep[G] v$. Then,
the \emph{$k$-approximation} $\Lang[k]{G}$ of $\Lang{G}$ is defined as 
\[ \Lang[k]{G} = \{w\in T^* \mid S\grammarsteps[G,k] w \}. \] For
$k\ge 1$, we use $\CF[k]$ to denote the class of languages of the form
$\Lang[k]{G}$ for context-free grammars $G$. The languages in $\CF_k$
are called \emph{index-$k$ context-free languages}. It will later be
convenient to let $\CF_0$ denote the regular languages. Moreover,
$\fiCF$ is the union $\bigcup_{k\ge 1} \CF[k]$. Its members are called
\emph{finite-index context-free languages}. Note that although clearly
$\Lang{G}=\bigcup_{k\ge 1} \Lang[k]{G}$ for every individual grammar
$G$, the class $\fiCF$ is strictly contained in $\CF$: Salomaa has
shown that $D_1^*$, the semi-Dyck language over one pair of parentheses,
is not contained in $\fiCF$~\cite{Salomaa1969}.

A \emph{$d$-dimensional (labeled) Petri net\footnote{This definition is closer to
what is known as a Vector Addition System, but these models are well-known to
be equivalent with respect to generated languages.}} is a tuple
$N=(X,E,\mu_0,F)$, where $X$ is an alphabet, $T$ is a finite subset of
$(X\cup\{\eword\})\times \Z^d$ whose elements are called \emph{transitions},
$\mu_0\in \N^d$ is the \emph{initial marking}, and $F\subseteq \N^d$ is a
finite set of \emph{final markings}.
For $\mu,\mu'\in\N^d$ and $w\in X^*$, we write $\mu\lautsteps[N]{w}\mu'$ if there are $\mu_0,\ldots,\mu_n\in\N^d$
 and transitions $(x_1,\nu_1),\ldots,(x_n,\nu_n)\in T$ such that $w=x_1\cdots x_n$, $\mu_0=\mu$, $\mu_n=\mu'$, and $\mu_i=\mu_{i-1}+\nu_i$ for $i\in\{1,\ldots,n\}$. Moreover, we define
\begin{align*}
\Lang{N,\mu,\mu'}=\{w\in X^* \mid \mu\lautsteps[N]{w} \mu'\}, && \Lang{N}=\bigcup_{\mu\in F} \Lang{N,\mu_0,\mu}.
\end{align*}
The language $\Lang{N}$ is said to be \emph{generated by $N$}.
A languages is a \emph{Petri net language} if it is generated by some labeled
Petri net. By $\Petri$, we denote the class of all Petri net languages.
Observe that we have $\Petri=\bigcup_{n\ge 0} \VA{\B^n}$.

The main result of this section involves the language class
$\fiCF\triointersect\Petri$.  We will use the fact that this is a full
trio, which follows from the following classical result.  See
\cite[Theorem~3.6.1]{Ginsburg1975} for a proof.
\begin{prop}\label{triointersecttrio} 
If $\C$ and $\D$ are full trios, then $\C\triointersect\D$ is a full trio as well.
\end{prop}

In our notation, the result of Atig and Ganty can be stated as follows.
\begin{thm}[Atig and Ganty~\cite{AtigGanty2011}]\label{atigganty}
The class $\fiCF\triointersect\Petri$ has a decidable emptiness problem.
\end{thm}

Here, we present a language class including both $\fiCF\triointersect\Petri$ and
$\VA{\MonDecidableSimple}$ where emptiness is still decidable. First, consider
the following hierarchy. 
Let 
\begin{align*} \HF_0=\Petri, && \text{$\HF_{i+1}=\HomSLI{\Alg{\HF_i}}$ for $i\ge 0$}, && \HF=\bigcup_{i\ge 0}\HF_i. \end{align*}
The class $\HF$ captures the expressive power of valence automata over monoids in $\MonDecidableSimple$:
\begin{prop}
$\VA{\MonDecidableSimple}=\HF$.
\end{prop}
\begin{proof}
For the inclusion ``$\subseteq$'', we prove that for every $M\in\MonDecidableSimple$, we have 
$\VA{M}\subseteq \HF_i$ for some $i\ge 0$. Clearly, we have $\VA{\B^n}\subseteq \HF_0$.
Moreover, if $\VA{M}\subseteq\HF_i$, then 
\[ \VA{M\times\Z}\subseteq\HomSLI{\VA{M}}\subseteq \HomSLI{\HF_i}\subseteq \HF_{i+1}, \]
in which the first inclusion follows from \cref{productz}.
Finally, if $\VA{M_0}\subseteq\HF_i$ and $\VA{M_1}\subseteq\HF_{j}$, then $\VA{M_0},\VA{M_1}\subseteq\HF_k$ for $k=\max\{i,j\}$ and thus
\[ \VA{M_0*M_1}\subseteq\Alg{\VA{M_0}\cup\VA{M_1}}\subseteq\Alg{\HF_k}\subseteq\HF_{k+1}, \]
where the first inclusion is due to \cref{freeproduct}. This completes the proof of the inclusion ``$\subseteq$''.

For the inclusion ``$\supseteq$'', we show by induction on $i$ that
$\HF_i\subseteq\VA{\MonDecidable}$ for every $i\ge 0$. Since
$\VA{\MonDecidable}=\VA{\MonDecidableSimple}$ by
\cref{emptiness:decidable:intuition}, this is sufficient.  Clearly,
the inclusion $\HF_0=\bigcup_{n\ge 0}\VA{\B^n}\subseteq
\VA{\MonDecidable}$ holds. Now suppose
$\HF_i\subseteq\VA{\MonDecidable}$ and let $L$ be a member of
$\HF_{i+1}=\HomSLI{\Alg{\HF_i}}$.  This means we have
$L=h(K\cap\ParikhInv{S})$ for some homomorphism $h$, a language $K$
from $\Alg{\HF_i}$, and a semilinear set $S$. As a member of
$\Alg{\HF_i}$, the language $K$ is generated by an $\HF_i$-grammar
$G$. Each right-hand side in $G$ is contained in $\HF_i$ and thus, by
induction, in $\VA{\MonDecidable}$. Hence, suppose the right-hand
sides of $G$ are $K_1,\ldots,K_n$ with $K_i\in\VA{M_i}$ for
$M_1,\ldots,M_n\in\MonDecidable$. Consider the monoid $M=M_1*\cdots
*M_n$.  Since each $M_i$ embeds into $M$, the languages
$K_1,\ldots,K_n$ belong to $\VA{M}$. Thus, $K$ is a member of
$\Alg{\VA{M}}$, which equals $\VA{\B*\B*M}$ according to \cref{algbb}.
According to \cref{productz}, this implies that $L$ belongs to
$\VA{(\B*\B*M)\times\Z^k}$ for some $k\ge 0$.  Since
$(\B*\B*M)\times\Z^k$ is a member of $\MonDecidable$, we know that $L$
belongs to $\VA{\MonDecidable}$. We have thus shown
$\HF_{i+1}\subseteq\VA{\MonDecidable}$, which establishes the
inclusion ``$\subseteq$''. \myqed
\end{proof}

Our new class is defined as follows. Let 
\begin{align*}\HG_0=\fiCF\triointersect\Petri, &&  \text{$\HG_{i+1}=\HomSLI{\Alg{\HG_i}}$~for $i\ge 0$}, && \HG=\bigcup_{i\ge 0}\HG_i.\end{align*}
Then clearly $\HF_i\subseteq \HG_i$ for $i\ge 0$ and hence
$\VA{\MonDecidableSimple}=\HF\subseteq \HG$.  Moreover, we obviously have
$\fiCF\triointersect\Petri\subseteq \HG$. 
We shall prove the following.
\begin{thm}\label{synthesisresult}
The class $\HG$ has a decidable emptiness problem.
\end{thm}

We show \cref{synthesisresult} by proving a slightly stronger version
of \cref{atigganty}: Atig and Ganty reduce the emptiness problem of
$\fiCF\triointersect\Petri$ to the emptiness problem for priority
multicounter machines. We strengthen this slightly and show that for
each language $L$ in $\fiCF\triointersect\Petri$, one can construct a
priority multicounter machine $\cA$ with
$\Parikh{\Lang{\cA}}=\Parikh{L}$, in other words:
$\Parikh{\fiCF\triointersect\Petri}\subseteq\Parikh{\Prio}$. This
allows us to apply \cref{algparikh} and \cref{prio:closure} to
conclude that $\Parikh{\HG}\subseteq\Parikh{\Prio}$.

The following observation provides a decomposition of languages in $\fiCF$.  A
context-free grammar $G=(N,T,P,S)$ is called \emph{linear} if every production
$A\to w$ in $G$ satisfies $|w|_N\le 1$. A language is called \emph{linear
context-free} if it is generated by a linear context-free grammar. Note that
$\CF_1$ is precisely the class of linear context-free languages.
\begin{prop}\label{decomplinear}
Suppose $k\ge 1$. A language belongs to $\CF_{k}$ if and only if it can be written
as $\sigma(L)$ for a linear context-free language $L$ and a
$\CF_{k-1}$-substitution $\sigma$.
\end{prop}
\begin{proof}
  We prove the statement by induction on $k$. For $k=1$, it
  essentially states that linear context-free languages are closed
  under regular substitutions, which is clearly true. For the
  induction step, we use a result of Atig and Ganty.  Let
  $G=(N,T,P,S)$ be a context-free grammar in CNF. For each $i\ge 0$,
  let $A^{[i]}$ be a fresh symbol.  For each $\ell\ge 0$, we define a
  grammar $G^{[\ell]}$ as follows. We have $G^{[\ell]}=(N^{[\ell]}, T,
  P^{[\ell]}, S^{[\ell]})$ with $N^{[\ell]}=\{ A^{[i]} \mid A\in
  N,~0\le i\le \ell\}$ and $P^{[\ell]}$ is the smallest set of
  productions such that
\begin{enumerate}
\item for each $A\to BC$ in $P$, we have $A^{[i]}\to B^{[i]}C^{[i-1]}$ and $A^{[i]}\to B^{[i-1]}C^{[i]}$ in $P^{[\ell]}$ for every index $i\in\{1,\ldots,\ell\}$,
\item for each $A\to w$ in $P$ with $w\in T\cup \{\eword\}$, we have $A^{[i]}\to w$ in $P^{[\ell]}$ for every $i\in\{0,\ldots,\ell\}$.
\end{enumerate}
For each nonterminal $A$ of $G$, we define
\begin{align*}
\Lang{G,A}=\{w\in T^* \mid A\grammarsteps[G] w \}, && \Lang[\ell]{G,A}=\{w\in T^* \mid A\grammarsteps[G,\ell] w \}.
\end{align*}
Atig and Ganty~\cite{AtigGanty2011} show that for every
$i\in\{0,\ldots,\ell\}$, one has
\begin{equation} \Lang{G^{[\ell]}, A^{[i]}}=\Lang[i+1]{G, A}. \label{ageq}\end{equation} 

Now suppose $K$ belongs to $\CF_{k+1}$ with $K=\Lang[k+1]{G}$ where $G$ is in
CNF. According to \eqref{ageq}, we have $\Lang{G^{[k]}}=K$. We now construct a
(linear) context-free grammar $G'=(N',T',P',S')$ (that is not necessarily in
CNF) as follows. It has terminal symbols $T'=T\cup \{A^{[i]} \mid 0\le i\le
k-1\}$ and its nonterminal symbols are $N'=\{A^{[k]} \mid A\in N\}$. As
productions, it contains all those productions of $G$ whose left-hand side
belongs to $N'$. Moreover, the substitution $\sigma\colon
T'^*\to\Powerset{T^*}$ is defined as follows: For $a\in T$, we set
$\sigma(a)=\{a\}$. For $A^{[i]}\in T'$, we define
$\sigma(A^{[i]})=\Lang{G^{[k]},A^{[i]}}$. Since for $A^{[i]}\in T'$, we have
$i\le k-1$, the equation \eqref{ageq} tells us that $\sigma(A^{[i]})$ belongs
to $\CF_{i+1}\subseteq \CF_k$.  Hence, $\sigma$ is a $\CF_{k}$-substitution. Moreover, an
inspection of the definition of $G^{[\ell]}$ yields that $G'$ is clearly
linear. Finally, we have $K=\sigma(\Lang{G'})$, so that with $L=\Lang{G'}$, we
have proven the ``only if'' direction of the \lcnamecref{decomplinear}.  The
other direction is obvious.\myqed
\end{proof}

We now turn to the key lemma (\cref{simulator}) of our slightly
stronger version of Atig and Ganty's result. We want to show that
given an index-$k$ context-free language $K$ and a Petri net
language $L$, one can construct a priority multicounter machine $\cA$
with $\Parikh{\Lang{\cA}}=\Parikh{K\cap L}$. The proof proceeds by
induction on the index $k$, which warrants a strengthening of the
statement. 

We need some terminology.  For a vector
$\mu=(m_1,\ldots,m_d)\in\N^d$ and $k\ge d$, we denote by
$\zerofill{\mu}$ the vector $(0,\ldots,0,m_1,\ldots,m_d)\in\N^k$.  The
dimension $k$ will always be clear from the context.
In order to make the induction work, we need to construct priority
counter machines with the additional property that for a particular
$d$, they never zero-test their $d$ topmost counters. Therefore, for a
priority $k$-counter machine $\cA=(Q,X,E,q_0,F)$ and $d\le k$, we define
$\cA_d$ to be the machine obtained from $\cA$ removing all transitions
$(q,x,\ell,\nu,q')$ with $\ell>d$. In other words, we remove all
transitions that perform a zero-test on a counter other than
$1,\ldots,d$.  We define the language
\[ \Lang[d]{\cA,q,\mu,q',\mu'}=\{w\in X^* \mid (q,\zerofill{\mu})\lautsteps[\cA_d]{w}(q',\zerofill{\mu'}) \}. \]

For a language $K\subseteq X^*$ and a $d$-dimensional Petri net
$N=(X,E,\mu_0,F)$, we say that a priority $k$-counter machine $\cA$ is a
\emph{$(K,N)$-simulator} if $k\ge d$ and there are two states $p$ and
$p'$ in $\cA$ such that for every $\mu,\mu'\in\N^d$, we have
\begin{equation} \Parikh{\Lang[d]{\cA,p,\mu,p',\mu'}}=\Parikh{K\cap \Lang{N,\mu,\mu'}}.\label{simequation}\end{equation}
In this case, $p$ and $p'$  are called \emph{source} and \emph{target}, respectively.

\begin{lem}\label{simulator}
Given a language $K$ in $\fiCF$ and a labeled Petri net $N$, one can construct
a $(K,N)$-simulator.
\end{lem}
\begin{proof}
Let $N=(X,E,\mu_0,F)$ be a $d$-dimensional Petri net and let $K$ belong to $\CF_k$.
We proceed by induction on $k$. If $k=0$, then $K$ is accepted by
some finite automaton $\cB$.  We may assume that $\cB$ has an initial state $p$ and
one final state $p'$.  One can construct a priority $d$-counter machine $\cA$ by a
product construction from $N$ and $\cB$ such that $\cA$ has the same state set as
$\cB$ and
\[ \Lang[d]{\cA,p,\mu,p',\mu'}=K\cap \Lang{N,\mu,\mu'}, \]
meaning it is indeed a $(K,N)$-simulator.

For the induction step, suppose $k\ge 1$. According to
\cref{decomplinear}, there is a linear context-free language
$L\subseteq Y^*$ and a $\CF_{k-1}$-substitution $\sigma\colon
Y\to\Powerset{X^*}$ with $K=\sigma(L)$. Let us begin with some
explanation.  Since $L$ is linear context-free, it is given by a
grammar $G=(\bar{N},Y,P,S)$ where every production is of the form
$A\to x_1Bx_2$ or $A\to\eword$ with $A,B\in\bar{N}$ and $x_1,x_2\in
Y\cup\{\eword\}$.  Let $D\subseteq P^*$ be the regular language of
production sequences that correspond to derivations in $G$ and let
$g_1, g_2\colon P^*\to Y^*$ be the morphisms where for $\pi=A\to
x_1Bx_2$ (with $x_1,x_2\in Y\cup\{\varepsilon\}$), we set
$g_i(\pi)=x_i$.  Then, we have $L=\{g_1(w)g_2(\rev{w}) \mid w\in D\}$.

Therefore, if $\tau_i$
is the $\CF_{k-1}$-substitution with $\tau_i(\pi)=\sigma(g_i(\pi))$ for
$i=1,2$, then $K=\sigma(L)$ consists of all words in
$\tau_1(w)\tau_2(\rev{w})$ for $w\in D$.
In other words, $K$ contains precisely those words of the form
\begin{equation} u_1\cdots u_n v_n\cdots v_1 \label{palindrome-run} \end{equation}
such that there is a word $w=\pi_1\cdots \pi_n\in D$, $\pi_1,\ldots,\pi_n\in P$
with $u_i\in\tau_1(\pi_i)$ and $v_i\in\tau_2(\pi_i)$ for $i\in\{1,\ldots,n\}$.

Our task is to construct a $(K,N)$-simulator $\cA'$. This means, using a
priority counter machine, we have to simulate---up to Parikh
image---all runs of $N$ with labels as in \cref{palindrome-run}.  By
induction, we have a $(\tau_i(\pi), N)$-simulator $\cA_{\pi,i}$ for each
$\pi\in P$ and $i\in\{1,2\}$. Each of the machines $\cA_{\pi,i}$ has
$\ge d$ counters, so we may clearly assume that for some $\ell\ge 0$,
they all have $\ell+d$ counters.  Moreover, by definition of a
$(\tau_i(\pi), N)$-simulator, these machines never perform a zero-test
on the $d$ top-most counters.  Moreover, using a zero-test, we can
guarantee that when $\cA_{\pi,i}$ reaches its target state, its first
$\ell$ counters are zero.

The basic idea is that $\cA'$ performs a run of an automaton for $D$,
which reads a word $w=\pi_1\cdots\pi_n$.  For each $j=1,\ldots,n$, it
executes a computation of $\cA_{\pi_j,1}$ (reading $u_j$) and a
computation of $\cA_{\pi_j,2}$ (reading $v_j$). Hence, $\cA'$ reads the
word $u_1v_1u_2v_2\cdots u_nv_n$, which is clearly Parikh-equivalent
to $u_1\cdots u_nv_n\cdots v_1$.

We have to make sure that all these runs of the machines $\cA_{\pi_j,i}$
are compatible in the sense that they can be executed in the order
prescribed by \cref{palindrome-run}. To this end, all the executions
of $\cA_{\pi_1,1},\ldots,\cA_{\pi_n,1}$ share one set of $\ell+d$
counters. The executions of $\cA_{\pi_1,2},\ldots,\cA_{\pi_n,2}$ also
share a set of counters, but they are executed backwards. The counters
for the backward execution are also $\ell+d$ many, but since each
execution of some $\cA_{\pi,i}$ leaves the first $\ell$ counters empty,
the forward and the backward simulation can share the first $\ell$
counters between them. This leaves us with $\ell+2d$ counters: We use
counters $1,\ldots,\ell+d$ to simulate $\cA_{\pi_j,1}$ and we use
counters $1,\ldots,\ell$ and $\ell+d+1,\ldots,\ell+2d$ to simulate
$\cA_{\pi_j,2}$ (backwards). Therefore, we call counters $1,\ldots,\ell$
\emph{auxiliary counters}, whereas the counters $\ell+1,\ldots,\ell+d$
are called \emph{forward counters}. The counters
$\ell+d+1,\ldots,\ell+2d$ are dubbed \emph{backward counters}.

In addition, we have to make sure that the executions of $N$
corresponding to $v_n\cdots v_1$ can be executed after the executions
corresponding to $u_1\cdots u_n$. Therefore, after executing the run
of the automaton for $D$, $\cA'$ simultaneously counts down the forward
and the backward counters and then performs a zero-test on the
counters $1,\ldots,\ell+2d$.

Finally, in order to be a $(K,N)$-simulator, $\cA'$ must have $d$
top-most counters so that the following holds: If we simulate the
computation $\mu\lautsteps[N]{u_1\cdots u_nv_n\cdots v_1} \mu'$ with
$\mu,\mu'\in\N^d$, then the $d$ top-most counters of $\cA'$ must contain
$\mu$ in the beginning and $\mu'$ in the end. To this end, we add an
additional set of $d$ counters, called \emph{global counters}. Hence,
in total, $\cA'$ has $\ell+3d$ counters:
\newcommand{\bracefont}[1]{\textrm{#1}}
\[ \underbrace{1,\ldots,\ell}_{\bracefont{auxiliary}},~\underbrace{\ell+1,\ldots,\ell+d}_{\bracefont{forward}},~\underbrace{\ell+d+1,\ldots,\ell+2d}_{\bracefont{backward}},~\underbrace{\ell+2d+1,\ldots,\ell+3d}_{\bracefont{global
}}. \]
The global counters are used as follows. The machine $\cA'$ starts with
counters $\zerofill{\mu}\in\N^{\ell+3d}$. First, it nondeterministically
subtracts some vector $\nu_1\in\N^d$ from the global counters and
simultaneously adds it to the forward counters. Then, it
nondeterministically adds a vector $\nu_2'\in\N^d$ to both the global
counters and the backward counters.  After performing the simulation
of the $\cA_{\pi_1,1},\ldots,\cA_{\pi_n,1}$ and the
$\cA_{\pi_1,2},\ldots,\cA_{\pi_n,2}$, suppose the forward counters contain
$\nu_1'\in\N^d$ and the backward counters contain $\nu_2\in\N^d$. As
described above, $\cA'$ afterwards compares the forward and backward
counters, ensuring that $\nu_1'=\nu_2$ and thus:
\[ \nu_1\lautsteps[N]{u_1\cdots u_n} \nu_1'=\nu_2\lautsteps[N]{v_n\cdots v_1} \nu'_2.\]
Observe that this guarantees that the global counters of $\cA'$ reflect
the counters of the simulated computation of $N$: In the end, they are
precisely $\zerofill{\mu'}\in\N^{\ell+3d}$, where $\mu'=\mu-\nu_1+\nu'_2$, which means
$\mu\lautsteps[N]{u_1\cdots u_nv_n\cdots v_1} \mu'$.

Let us make the description of $\cA'$ more precise.
\begin{enumerate}[label=(\roman*)]
\item $\cA'$ has a state $p$, where it nondeterministically subtracts
  tokens from the global counters and simultaneously adds them to the
  forward counters.
\item Note that $D\subseteq P^*$ can be accepted by a finite automaton
  with state set $\bar{N}$, the non-terminals of $G$. Therefore, from
  the state $p$, $\cA'$ can enter the state $S\in\bar{N}$ to start
  simulating the automaton for $D$.
\item In a state $A\in\bar{N}$, $\cA'$ selects a production $\pi=A\to
  x_1Bx_2 \in P$ and then executes a computation of $\cA_{\pi,1}$ in the
  auxiliary and the forward counters.  Then, it executes a computation
  of $\cA_{\pi,2}$ backwards on the auxiliary and backward
  counters. Then, it switches to state $B\in\bar{N}$.
\item If $\cA'$ is in state $A\in\bar{N}$ and there is a production
  $A\to\eword\in P$, then $\cA'$ switches to a state $p''$, in which it
  simultaneously counts down the forward and the backward counters.
  From $p''$ it non-deterministically switches to $p'$ while
  performing a zero-test on all counters $1,\ldots,\ell+2d$.
\end{enumerate}
In conclusion, it is clear that for $\mu,\mu'\in\N^d$, we have
$(p,\zerofill{\mu})\lautsteps[\cA']{w} (p',\zerofill{\mu'})$ if and only
if $w=u_1v_1u_2v_2\cdots u_nv_n$ such that there is a word
$\pi_1\cdots \pi_n\in D$, $\pi_1,\ldots,\pi_n\in P$, such that
$u_j\in\tau_1(\pi_j)$ and $v_j\in\tau_2(\pi_j)$. Thus, $\cA'$ is a
$(K,N)$-simulator.\myqed

\end{proof}

We are now ready to prove the slightly stronger version of the decidability
result of Atig and Ganty.
\begin{thm}\label{intersection}
Given $K$ in $\fiCF\triointersect \Petri$, one can construct a priority multicounter
machine $\cA$ with $\Parikh{\Lang{\cA}}=\Parikh{K}$.
\end{thm}
\begin{proof}
Since the languages of priority multicounter machines are closed under
morphisms, we may assume that $K=C\cap P$, where $C$ is in $\CF_k$ and
$P=\Lang{N}$ for a $d$-dimensional labeled Petri net $N=(X,T,\mu_0,F)$.
\Cref{simulator} allows us to construct a $(C,N)$-simulator $\cA$ with source $p$
and target $p'$. This means, for each $\mu\in\N^d$, we have
$\Parikh{\Lang[d]{\cA,p,\mu_0,p',\mu}}=\Parikh{C\cap \Lang{N,\mu_0,\mu}}$ and in
particular
\[ \Parikh{\underbrace{\bigcup_{\mu\in F} \Lang[d]{\cA,p,\mu_0,p',\mu}}_{=:L}}=\Parikh{\bigcup_{\mu\in F}C\cap \Lang{N,\mu_0,\mu}}=\Parikh{C\cap P}. \]
Since we can clearly construct a priority multicounter machine for $L$, the
proof of the \lcnamecref{intersection} is complete.\myqed
\end{proof}

This allows us to prove \cref{synthesisresult}.
\begin{proof}[\cref{synthesisresult}]
Given a language in $\HG_i$, we can recursively construct a Parikh equivalent
priority multicounter machine.  According to \cref{intersection}, this is true
of $\HG_0=\fiCF\triointersect\Petri$. Furthermore, \cref{algparikh} and
\cref{prio:closure} tell us that if we can carry out such a construction for
$\HG_i$, we can also do it for $\HG_{i+1}$.\myqed
\end{proof}

\section{Conclusion} 
Of course, an intriguing open question is whether the storage
mechanisms corresponding to $\MonRemainingSimple$ have a decidable
reachability problem.  First, since their simplest instance are
pushdown Petri nets, this extends the open question concerning the
latter's reachability.  Second, they naturally subsume the priority
multicounter machines of Reinhardt. This makes them a candidate for
being a quite powerful model for which reachability might be
decidable.

Observe that if these storage mechanisms turn out to exhibit decidability, this
would mean that the characterization of Lohrey and Steinberg 
(\cref{ratmp:graphgroups}) remains true for all graph monoids. This can be
interpreted as evidence for decidability.

\paragraph{Acknowledgments} The author is grateful to the anonymous
referees of both the conference and the journal version.  Their
helpful comments have greatly improved the presentation of this work.

\bibliographystyle{abbrv}
\bibliography{bibliography} 

\end{document}